\documentclass[conference]{IEEEtran}
\IEEEoverridecommandlockouts
\usepackage{cite}
\usepackage{amsmath,amssymb,amsfonts}
\usepackage{algorithmic}
\usepackage{graphicx}
\usepackage{textcomp}
\usepackage{xcolor}
\usepackage{comment}
\def\BibTeX{{\rm B\kern-.05em{\sc i\kern-.025em b}\kern-.08em
    T\kern-.1667em\lower.7ex\hbox{E}\kern-.125emX}}

\usepackage{enumitem}
\setlist[description]{leftmargin=0.05cm,labelindent=0.275cm}

\usepackage{multirow}
\usepackage[hidelinks]{hyperref}
\usepackage{subcaption}
\usepackage{caption}
\usepackage[algo2e,linesnumbered,ruled,vlined,boxed]{algorithm2e} 

\usepackage{amssymb,amsmath,amsthm}

\newtheorem{theorem}{Theorem}
\newtheorem{lemma}{Lemma}
\newtheorem{proposition}{Proposition}

\usepackage{dsfont}

\usepackage{todonotes,soul}

\newcommand{\Vard}[2]{{\normalfont \text{Var}}_{#2} \left[ #1 \right]}

\newcommand{\ind}[1]{\mathds{1} \left[ #1 \right] }

\newcommand\numberthis{\addtocounter{equation}{1}\tag{\theequation}}

\DeclareMathOperator*{\argmin}{arg\,min}
\DeclareMathOperator*{\argmax}{arg\,max}

\newcommand{\Q}{\mathbb{Q}}
\newcommand{\R}{\mathbb{R}}

\newcommand{\calS}{\mathcal{S}}

\newcommand{\bigO}{\mathcal{O}}

\newcommand{\sample}{\mathcal{S}}

\newcommand{\E}{\mathbb{E}}

\newcommand{\samplalg}{\textsc{ImportanceSampler}}
\newcommand{\percis}{\textsc{PercIS}}
\newcommand{\unif}{\textsc{UNIF}}
\newcommand{\intern}{I}

\newtheorem{definition}{Definition}

\newlength{\commentWidth}
\let\oldnl\nl
\newcommand{\nonl}{\renewcommand{\nl}{\let\nl\oldnl}}
\DontPrintSemicolon
\SetKwInOut{Input}{Input}\SetKwInOut{Output}{Output}
\usepackage{algorithmic}

\usepackage{xspace}
\usepackage{booktabs}

\usepackage{hyperref}
\usepackage{mathtools}
\usepackage{amsmath}
\usepackage{amsfonts}
\usepackage{bm}
\usepackage{dsfont}
\usepackage{upgreek}

\usepackage[hidelinks]{hyperref}
\usepackage{xr-hyper}  
\newif\ifwithappendix
\withappendixtrue

\begin{document}


\title{Fast Percolation Centrality Approximation \\ with Importance Sampling}

\author{

\IEEEauthorblockN{Antonio Cruciani}
\IEEEauthorblockA{\textit{Department of Computer Science} \\
\textit{Aalto University}\\
Espoo, Finland \\
\url{antonio.cruciani@aalto.fi}}
\and
\IEEEauthorblockN{Leonardo Pellegrina}
\IEEEauthorblockA{\textit{Department of Information Engineering} \\
\textit{University of Padova}\\
Padova, Italy \\
\url{leonardo.pellegrina@unipd.it}}
}

\maketitle

\begin{abstract}
In this work we present {\sc PercIS}, an algorithm based on Importance Sampling to approximate the percolation centrality of all the nodes of a graph. 
Percolation centrality is a generalization of betweenness centrality to attributed graphs, and is a useful measure to quantify the importance of the vertices in a contagious process or to diffuse information. 
However, it is impractical to compute it exactly on modern-sized networks. 

First, we highlight key limitations of state-of-the-art sampling-based approximation methods for the percolation centrality, 
showing that in most cases they cannot achieve accurate solutions efficiently. 
Then, we propose and analyze a novel sampling algorithm based on Importance Sampling,
proving tight sample size bounds to achieve high-quality approximations. 

Our extensive experimental evaluation shows that {\sc PercIS} computes high-quality estimates and scales to large real-world networks, while significantly outperforming, in terms of sample sizes, accuracy and running times, the state-of-the-art. 
\end{abstract}



\section{Introduction}
Identifying important nodes in a graph is a key task in graph mining. 
The most common technique for this task is to use a centrality measure~\cite{newman2018networks}, 
such as the popular PageRank~\cite{page1999pagerank}, 
betweenness 
~\cite{anthonisse1971rush,freeman1977set}, 
and closeness~\cite{bavelas1948mathematical} centralities. 
These measures quantify the importance of nodes under different perspectives. 
For instance, 
the betweenness centrality of a node is defined in terms of the fraction of shortest paths that traverse it, 
while closeness centrality gives more importance to nodes that are close to all other nodes of the network. 
While the type of measure to use typically depends on the application~\cite{ghosh2014interplay},
centrality measures are important for several tasks, such as 
discovering communities~\cite{girvan2002community} and vulnerabilities~\cite{iyer2013attack}, 
or to shape information diffusion over the network~\cite{kempe2003maximizing}. 

In this paper we consider the \emph{percolation centrality}~\cite{Piraveenan_2013}, a measure useful in settings where graphs model a contagious process in a network (e.g., the spread of infection in a population or  misinformation over a social network). 
The percolation centrality is a generalization of the betweenness centrality to attributed graphs:
it quantifies the importance of a node $v$ in terms of the \emph{weighted} fraction of shortest paths that traverse $v$. 
The weight of each shortest path of the graph 
depends on 
the difference of the \emph{percolation states} of the terminals of the path; 
the percolation state of a node quantify its level of contamination. 
Intuitively, percolation centrality measures the importance of a node in terms of its potential contribution to the diffusion of a contagion, i.e., when it connects nodes with high percolation state with nodes with low percolation state. 

The study of percolation processes has been introduced by~\cite{Broadbent_1957} to model the passage of a fluid in a medium. 
Subsequently, Piraveenan et al.~\cite{Piraveenan_2013} proposed the percolation centrality measure for the nodes of a network, in which the medium are the vertices of a given graph $G$ and each vertex $v$ in $G$ has a real-valued percolation state $x_v \in [0,1]$ that reflects the level of contamination of the node $v$. 

Similarly to the betweenness centrality, 
a technique to compute the percolation centrality scores of all the nodes in a graph $G$ is to solve the All Pair Shortest Paths (APSP) problem (e.g., to run a BFS or Dijkstra algorithm from each node $v$). 
Unfortunately, under the APSP conjecture~\cite{Abboud_2014} no truly subcubic time ($\bigO(n^{3-\varepsilon})$ for any $\varepsilon >0$) algorithm can be designed. 
For the betweenness centrality, faster methods have been developed, 
such as Brandes' algorithm~\cite{brandes2001faster},  
but they still feature analogous lower bounds to their complexities~\cite{borassi2016into}. 
Since applications require to analyze large graphs, e.g., with hundreds of millions of nodes and edges, exact computation of such scores is clearly prohibitive. 
Therefore, the only viable solution is to provide efficient-to-compute \emph{approximations} with high-quality accuracy guarantees. 
The main challenge is to tightly relate the quality of the approximation with the cost of computing it, i.e., the running time of the approximation algorithm. 
This is the main goal of our work. 

Lima et al.~\cite{Lima_2020,Lima_2022} generalized the techniques proposed by Riondato and Kornaropoulos~\cite{Riondato_2016} and Riondato and Upfal~\cite{Riondato_2018} for the betweenness centrality to design methods for approximating the percolation centrality. 
The high-level idea is to randomly sample shortest paths of the graph, and use the (weighted) fraction of the paths that traverse $v$ as an estimate of its percolation centrality. 
The main technical challenge is to bound the number of random samples required to achieve an accurate approximation for all nodes of the graph; 
since generating each sample is expensive, 
as it requires exploring the graph to compute and sample shortest paths, 
such bound directly impacts the efficiency of the approximation algorithm. 
While~\cite{Lima_2020,Lima_2022} derive sample size bounds with elegant tools from statistical learning theory, such as 
pseudodimension~\cite{pollard2012convergence} (a generalization of the VC-dimension~\cite{Vapnik:1971aa}) 
and Rademacher Averages~\cite{Barlett_2003,pellegrina2022mcrapper}, we highlight some technical issues that prevent these methods to be useful in practical applications (as we discuss in Section~\ref{sec:percapproxprelims}). 
For this reason, no truly effective algorithm exists to approximate the percolation centrality. 

In this work we develop a new algorithm, that we call \percis, which leverages a refined random sampling distribution based on the Importance Sampling framework. 
We theoretically and empirically prove that \percis\ is the first truly effective method to compute high-quality approximations of percolation centrality. 
We now summarize our contributions:

\begin{itemize}
    \item we introduce a new estimator of the percolation centrality scores of all nodes of the graph based on a novel Importance Sampling distribution over the shortest paths of the graph. 
    We theoretically prove that this refined sampling distribution is able to drastically reduce the variance of the estimates, compared to previous approaches. 
    In fact, we highlight technical issues with previous sampling-based methods, that instead use a \emph{uniform} sampling distribution; we prove that such approaches are impractical, 
    by establishing strong \emph{lower bounds} to the number of required random samples. 
    
    \item leveraging the Importance Sampling distribution, we develop the \percis\ algorithm. 
    Our algorithm uses an efficient procedure to sample from the importance distribution, and 
     a two-phases sampling scheme to calibrate the number of random samples needed to achieve an high quality approximation of the centrality scores. 
     The main technical tool we develop is a new bound on the sufficient number of samples to approximate the
    percolation centrality for all nodes, that is governed by key parameters, such as the maximum variance of the percolation centrality estimators. 
    We show that this bound is much tighter than the ones from previous works. 

	\item  we perform an extensive experimental evaluation showing that our algorithm significantly improves on the state-of-the-art in terms of sample sizes, running time, and accuracy of the estimates.
\end{itemize}

\textit{Related Works}. 
Several methods have been developed to compute accurate approximations of centrality scores. 
The most relevant to our work are the ones tailored to compute high-quality estimates of the 
betweenness centrality~\cite{Riondato_2016,Riondato_2018,borassi2019kadabra,Cousins_2023,Pellegrina_2023}. 
As discussed previously, these approaches are based on elegant and advanced technical tools to establish their theoretical guarantees, and scale easily to large networks. 
For instance, \cite{borassi2019kadabra} leverages adaptive sampling, while  \cite{Pellegrina_2023} used progressive sampling and Monte Carlo Rademacher Averages to compute tight non-uniform approximation bounds. 
However, despite their effectiveness, such methods cannot be applied to approximating the percolation centrality, since the two measures differ substantially. 

Other works 
considered scalable approaches~\cite{boldi2013core,chechik2015average} to approximate different types of centralities, such as closeness and harmonic~\cite{marchiori2000harmony}. 
Recent methods proposed extensions to different settings, such as 
uncertain~\cite{saha2021shortest}, 
temporal~\cite{santoro2022onbra,zhang2024efficient,cruciani2024mantra,brunelli2024making}, 
heterogeneous~\cite{wang2024efficient}, 
and 
dynamic graphs~\cite{bergamini2014approximating,bergamini2015fully,hayashi2015fully}. 

Other methods tackle different problems. 
\cite{yoshida2014almost,mahmoody2016scalable,pellegrina2023efficient,bergamini2018scaling,angriman2021group} focus on centrality maximization, that aims at identifying the most central \emph{set of nodes} of cardinality at most $k$. 
\cite{bergamini2018improving,miyauchi2024local} seek to add or remove $k$ edges to the graph so to maximally increase, or decrease, the centrality of a node. 

Similarly to other previous works, our algorithm relies on an upper bound to the \emph{vertex diameter}, that is the maximum number of nodes that are internal to a shortest path. 
For this problem, several efficient bounding methods have been proposed~\cite{magnien2009fast,crescenzi2012computing,hayashi2015fully,ceccarello2020distributed}.

\section{Preliminaries}
We now introduce the notation and most relevant definitions that we use as the
groundwork of our proposed algorithm. 

\subsection{Graphs and Percolation centrality}

Let a (directed) graph $G= (V,E)$ with $n=|V|$ nodes and $m=|E|$ edges.
For any pair of nodes $s,t\in V$ we define $\Gamma_{st}$ to be the set of all the shortest paths from $s$ to $t$, and $\sigma_{st} = |\Gamma_{st}|$. 
For a given path $\tau_{st}\in\Gamma_{st}$, 
we define $\intern(\tau_{st})$ as the set of internal vertices of the shortest path $\tau_{st}$, i.e., $\intern(\tau_{st}) = \{v\in \tau_{st} ,  s \neq v \neq t \}$. 
A node $v$ is internal to $\tau_{st}$ when $v \in \intern(\tau_{st})$, 
i.e., when $\tau_{st}$ passes through $v$ and $s \neq v \neq t$.
Moreover, we define $\sigma_{st}(v) = |\{ \tau_{st} : v \in \intern(\tau_{st})  \}|$ as the number of shortest paths from $s$ to $t$ such that $v$ is internal. 
We then denote $D \geq \max_\tau | \intern(\tau) |$ as an upper bound to the \emph{vertex diameter} of $G$, i.e., the maximum number of nodes that are internal to any shortest path of the graph. 

For every node $v \in V$, let $x_v \in [0,1]$ be the \emph{percolation state} of $v$. 
W.l.o.g. (up to a relabelling of the nodes), we assume that the percolation states are sorted in non-increasing order: $x_1 \geq x_2 \geq \dots \geq x_n$. 
Let $R(x) = \max(0,x)$ be the ramp function. 
 The percolation centrality $p(v)$ of the node $v$ is defined as follows~\cite{Piraveenan_2013}.
\begin{definition}
\label{def:percolationcentr}
Given a graph $G$ and the percolation states $\{ x_v , v \in V \}$, the percolation centrality $p(v)$ of a node $v\in V$ is defined as
\vspace{-1px}
\begin{align*}
   p(v) & =\sum_{\substack{s,t\in V\\ s \neq v \neq t}}{\frac{\sigma_{st}(v)}{\sigma_{st}} \kappa(s,t,v)} ,\\
\text{ where } ~ & \kappa(s,t,v) = \frac{R(x_s-x_t)}{\sum_{\substack{u,w\in V \\ u \neq v \neq w}}R(x_u-x_w)} . \numberthis \label{eq:kappaperc}
\end{align*}
\end{definition}
We note that the percolation centrality is \emph{normalized}, i.e., insensitive to the graph size, thus it holds $p(v) \in [0,1]$. 
This follows from 
$\kappa(s,t,v)\in [0,1]$ and  
$\sum_{\substack{s,t\in V \\ s \neq v \neq t}} \kappa(s,t,v) = 1 $.
We note that, for the percolation centrality $p(v)$ to be well defined, there must exist two nodes $s,t$ with $s \neq v \neq t$ such that $x_s \neq x_t$ (otherwise, the denominator in Eq.~\eqref{eq:kappaperc} is $0$).

\subsection{Percolation Centrality Approximation}
\label{sec:percapproxprelims}

Since computing the percolation centrality exactly is too expensive, our goal is to instead compute an accurate \emph{approximation} that is efficient to obtain. 
For a given \emph{accuracy} parameter $\varepsilon \in (0,1]$, provided in input by the user, we aim to return approximated values $\{ \tilde{p}(v) , v \in V \}$ that are within $\varepsilon$ of the exact values $\{ p(v) , v \in V \}$, for all nodes of the graph. 
We formalize this requirement with the \emph{$\varepsilon$-approximation}, which is defined as follows.
\begin{definition}
\label{def:epsapprox}
Given the accuracy parameter $\varepsilon \in (0,1]$, an $\varepsilon$-approximation $\{ \tilde{p}(v) , v \in V \}$ of the percolation centralities $\{ p(v) , v \in V \}$ satisfies 
\begin{align*}
| \tilde{p}(v) - p(v) | \leq \varepsilon , \forall v \in V .
\end{align*}
\end{definition}
When the $\varepsilon$-approximation is obtained using \emph{random sampling},  as in our case, the goal is to return an $\varepsilon$-approximation with high confidence, i.e., with probability at least $1 - \delta$, for $\delta \in (0,1)$. 
We now describe the sampling distribution and estimators employed by previous works. 

Lima et al.~\cite{Lima_2020} proposed a generalization of the approximation algorithm developed by Riondato and Kornaropoulos~\cite{Riondato_2016} for the betweenness centrality. 
The method of~\cite{Lima_2020} samples, uniformly at random, $\ell = \mathcal{O}(\ln(D/\delta)/\varepsilon^2)$ shortest paths of the graph, and approximates the percolation centrality of the node $v$ as the (weighted) fraction of shortest paths where $v$ is internal to, where the weights are given by the function $\kappa$. 
Each sample is generated by choosing two nodes $s,t$ uniformly at random, by computing the set of shortest paths $\Gamma_{st}$ from $s$ to $t$, e.g., with a BFS, and by choosing one shortest path uniformly at random from the set $\Gamma_{st}$. 
More precisely, by following this procedure a shortest path $\tau_{st} \in \Gamma_{st}$ is sampled with probability 
$(n(n-1) \sigma_{st})^{-1}$. 
After drawing a sample $\sample = \{ \tau^1 , \tau^2 , \dots , \tau^\ell \}$ of $\ell$ i.i.d. shortest paths, 
the returned estimate $\tilde{p}^*(v)$ for all $v \in V$ 
is the weighted average
$\tilde{p}^*(v) = \frac{1}{\ell} \sum_{i=1}^{\ell} \kappa(s,t,v) \ind{v \in \intern ( \tau^i_{st} ) } $. 
Subsequently, \cite{Lima_2022} proposed an estimator 
that considers
all the shortest paths $\Gamma_{st}$ from the sampled pairs $s,t$, similarly to the approach developed by Riondato and Upfal~\cite{Riondato_2018} for the betweenness.

We remark that, while the notion of $\varepsilon$-approximation defined in previous works~\cite{Lima_2020,Lima_2022} is similar to one we consider, 
the analyses carried out by~\cite{Lima_2020,Lima_2022} feature a key issue that makes the obtained theoretical guarantees vacuous. 
In fact, Lima et al.~\cite{Lima_2020,Lima_2022} focus on approximating the \emph{doubly normalized score} $p^*(v) = \frac{p(v)}{n(n-1)}$, instead of the standard percolation centrality $p(v)$ from Definition~\ref{def:percolationcentr} (e.g., see Definition 2.1 in~\cite{Lima_2020}). 
Note that the expected value $\E[\tilde{p}^*(v)]$ of the estimator 
$\tilde{p}^*(v)$ defined above
is $p^*(v)$,
and not $p(v)$. 
It is crucial to note that, for the same accuracy level $\varepsilon$,
an $\varepsilon$-approximation of $\{ p^*(v) , v \in V \}$ 
is significantly less informative than 
an $\varepsilon$-approximation of $\{ p(v) , v \in V \}$; 
in fact, 
any choice of $\varepsilon \geq \frac{1}{n(n-1)}$ 
leads to uninformative results w.r.t. $\{ p^*(v) , v \in V \}$, 
since the error bound $\varepsilon$ would be \emph{larger} than all centrality scores (observe that it holds $0 \leq p^*(v) \leq \frac{1}{n(n-1)}$, for all $v$).\footnote{We note that approximating $p^*(v)$ is trivial when $\varepsilon \geq \frac{1}{n(n-1)}$; this is due to the fact that $p^*(v) \in [0, \frac{1}{n(n-1)}]$, so simply reporting the values $\{ \tilde{p}^*(v) = 0 , v \in V \}$ is sufficient to guarantee $|\tilde{p}^*(v) - p^*(v)| \leq \varepsilon$ and to obtain an $\varepsilon$-approximation.} 
On the other hand, setting $\varepsilon < \frac{1}{n(n-1)}$ is unfeasible, since the analysis of~\cite{Lima_2020,Lima_2022}
implies that $\ell = \Omega(n^4)$ random samples are needed, which is far more expensive than running the exact algorithm. 
The cause for this issue arises from the use of \emph{uniform} distribution when sampling shortest paths from the graph; we prove, both theoretically and empirically, that this choice typically leads to poor results, due to high variance of the estimates, and a running time that is comparable to running the exact algorithm. 
In fact, this strategy fails to provide good approximations of the percolation centrality on realistic instances. 
Instead, our approach leverages \emph{Importance Sampling} and a more refined \emph{non-uniform distribution} over shortest paths; 
we prove that this is a key requirement to achieve an efficient approximation of the percolation centralities.

\subsection{Importance Sampling}
\label{sec:importancesampling}
In this section we introduce the Importance Sampling technique, 
a general method for Monte Carlo estimation problems~\cite{kahn1951estimation,rubinstein2016simulation,boucheron2013concentration,mitzenmacher2017probability}. 
An important application is variance reduction, i.e., to improve the accuracy of 
Monte Carlo approximations. 
Suppose we are given a discrete random variable $X$, and our goal is to approximate its expectation $\mu = \E_y[X]$ w.r.t. to a probability distribution $y$, defined as
$\E_y[X] = \sum_x x y(x) $,
where $y(x)$ denotes $\Pr_y(X=x)$. 
The standard approach is to draw i.i.d. random samples $\{ x_1 , \dots , x_\ell \}$ from $y$, and to approximate $\E_y[X]$ with the empirical mean $\tilde{\mu}_y = \frac{1}{\ell} \sum_{i=1}^\ell x_i$ of the observed values of $X$. 
However, the variance $\E_y[( X - \mu )^2]$ of this estimate may be large, due to the skew of the distribution $y$; 
furthermore, 
it may be difficult or very expensive to sample from $y$. 
Importance Sampling leverages an \emph{alternative} distribution $q$, also called \emph{importance} distribution, to improve the quality of the estimate by reducing the estimation variance;
the idea is to boost the sampling probability of important outcomes for the estimation of $\mu$. 
The alternative estimator $\tilde{\mu}_q$ uses i.i.d. random samples taken according to $q$, and is defined as 
$\tilde{\mu}_q = \frac{1}{\ell} \sum_{i=1}^\ell x_i \frac{y(x_i)}{q(x_i)} $.
A key requirement of the distribution $q$ is that, for all $x$ with $x y(x) > 0$, it must hold $q(x) > 0$.
The estimator $\tilde{\mu}_q$ is still unbiased, since $\E_q[\tilde{\mu}_q] = \mu$; depending on the choice of $q$, its variance $\E_q[( X \frac{y(X)}{q(X)} - \mu )^2]$ may be much smaller w.r.t. $\tilde{\mu}_y$, thus leading to more accurate estimates of the target mean $\mu$. 
A quantity often used to evaluate the quality of $q$ is the \emph{likelihood ratio} $\hat{d}$, which is 
$\hat{d} = \max_{x : q(x) > 0} \frac{y(x)}{q(x)} $. 

It is important to note that a bad choice of $q$, and a large likelihood ratio, may lead to a large estimation variance for the estimator $\tilde{\mu}_q$. 
In the following sections we will design a novel random sampling method that leverages Importance Sampling for the approximation of percolation centrality.

\section{Fast Estimation of the Percolation Centrality with \percis}
In this section we present our contributions in detail.
First, in Section~\ref{sec:isdistest} we define the Importance Sampling distribution and the estimator for accurately approximating the percolation centrality of all nodes of a graph. 
Then, in Section~\ref{sec:percisalg} we present our algorithm \percis, which is based on an efficient procedure to sample from the importance distribution, and a two-steps sampling scheme to achieve sharp data-dependent estimates. 
In Section~\ref{sec:percisanalysis} we prove the correctness of \percis, i.e., we prove that it outputs an $\varepsilon$-approximation with high probability. 
Then, in Section~\ref{sec:percisanalysiscomparison} we theoretically compare the performance and guarantees of \percis\ with the standard approach based on the uniform distribution. 

Due to space constraints, the proofs of our results, and the pseudocode of some subroutines, are deferred 
\ifwithappendix
to the appendix. 
\else
to the appendix in the online extended version~\cite{percisarxiv}. 
\fi

\subsection{Importance Sampling distribution and Estimator}
\label{sec:isdistest}
In this section we propose an importance distribution $q$ to accurately and efficiently estimate the percolation centralities of all nodes of a graph. 

We first observe that, as discussed in Section~\ref{sec:importancesampling} and from the definition of $p(v)$ given in Definition~\ref{def:percolationcentr}, 
an effective 
importance distribution $q$ for estimating the unknown mean $\mu_v = p(v)$ should boost the probability of sampling shortest paths connecting the nodes $s,t$ that have a high weight $\kappa(s,t,v)$. 
However, differently from the standard setting discussed in Section~\ref{sec:importancesampling}, our goal is to approximate a set of unknown means $\{ \mu_v , v \in V \}$, i.e., the centralities of \emph{all nodes} of the graph, instead of an individual expectation $\mu$. 
Moreover, the weights $\kappa(s,t,v)$ depend on $v$: each node $v$ defines a different probability distribution over pairs of nodes $s,t$. 
Therefore, it is not immediately clear how to define a proper, and accurate, importance distribution $q$ that allow approximating $p(v)$ for all nodes $v \in V$ of the graph. 

To address this issue, we define $\tilde{\kappa} : V \times V \rightarrow [0,1] $ as
\begin{align*}
\tilde{\kappa}(s,t) = \frac{R(x_s-x_t)}{\sum_{(u,w)\in V\times V}R(x_u-x_w)} .
\end{align*}
Intuitively, $\tilde{\kappa}(s,t)$ is similar to the percolation weight $\kappa(s,t,v)$, but without the dependence on $v$. 
Thus, our key intuition is that the weights $\tilde{\kappa}(s,t)$ define an accurate distribution over all pairs of nodes that,
under mild assumptions on the percolation states, 
allow an accurate estimation of the percolation centrality $p(v)$ \emph{simultaneously} for all nodes $v$ of the graph. 
For any shortest path $\tau_{st}$ of the graph from the node $s$ to $t$, 
we define the importance distribution $q(\tau_{st}) = \frac{\tilde{\kappa}(s,t)}{\sigma_{st}}$. 
To sample from $q$, we use the following procedure, described at high level:
sample two nodes $s,t$ with probability $\tilde{\kappa}(s,t)$;
compute the set of shortest paths $\Gamma_{st}$ from $s$ to $t$, and choose one shortest path uniformly at random from $\Gamma_{st}$. 
Let $\sample$ be a sample $\sample = \{ \tau^1 , \tau^2 , \dots , \tau^\ell \}$ of $\ell$ shortest paths drawn i.i.d. from $q$ following this procedure; 
the estimator $\tilde{p}(v)$ of $p(v)$ is  
\begin{align*}
\tilde{p}(v) = \frac{1}{\ell} \sum_{i=1}^{\ell} \frac{\kappa(s,t,v)}{\tilde{\kappa}(s,t)} \ind{v \in \intern(\tau^i_{st})} . 
\end{align*}
In Section~\ref{sec:percisalg} we describe our algorithm \percis, that implements this procedure efficiently. 
We now prove that $q$ is a valid importance distribution (i.e., that the expectation of $\tilde{p}(v)$ is well defined), 
and that 
the estimator $\tilde{p}(v)$ is unbiased. 
\begin{lemma}
\label{lemma:unbiasedest}
It holds $\E_q[\tilde{p}(v)] = p(v)$.
\end{lemma}
We now define the likelihood ratio $d_v$ for the node $v$ as
\begin{align*}
d_v = \max_{s,t \in V : \tilde{\kappa}(s,t) > 0}\frac{\kappa(s,t,v)}{\tilde{\kappa}(s,t)} ,  
\end{align*}
and the likelihood ratio $\hat{d}$ as $\hat{d} = \max_v d_v$. 
It is immediate to observe that $\tilde{p}(v)$ is an average of $\ell$ random variables with codomain $[0 , \hat{d}]$. 
Consequently, the likelihood ratio $\hat{d}$ also controls the variance of $\tilde{p}(v)$, for all nodes $v \in V$.
\begin{lemma}
\label{lemma:variancebound}
It holds 
$\Vard{\tilde{p}(v)}{q} \leq p(v)(\hat{d} - p(v)) \leq \hat{d} p(v)$.
\end{lemma}
As we will prove in Section~\ref{sec:percisanalysiscomparison}, the likelihood ratio $\hat{d}$ is small under extremely mild assumptions on the values of the percolation states. 
This guarantees that $\tilde{p}(v)$ is sharply concentrated towards its expectation $p(v)$, for all $v \in V$, as we prove in the following sections.

\subsection{\percis\ Algorithm}
\label{sec:percisalg}

In this section we introduce our new algorithm \percis, which is based on the importance distribution $q$ and the estimator introduced in the previous section. 
The pseudocode of \percis\ is described by Algorithm~\ref{algo:mainalg}.

\begin{algorithm2e}[htb!]
	\caption{\percis}\label{algo:mainalg}
	\KwIn{Graph $G=(V,E)$, percolation states $x_1 , x_2 , \dots , x_n$, $\ell_1 \geq 2$, $\varepsilon , \delta \in (0,1)$. }
	\KwOut{$\varepsilon$-approximation of $\{ p(v) , v \in V \}$ with probability $\geq 1- \delta$}
	
	$D \gets \textsc{VertexDiamUB}(G)$;
	
	$\sample \gets \text{\samplalg}(G, \{ x_v \} , \ell_1)$; \label{alg:firstsamplingstart}
	
	\lForAll{$v \in V$}{$\tilde{p}(v) \gets \frac{1}{\ell} \sum_{i=1}^{\ell} \frac{\kappa(s,t,v)}{\tilde{\kappa}(s,t)} \ind{v \in \intern( \tau^i_{st}) }$}
	
	$\hat{\rho} \gets \tilde{\rho}(\sample) + \sqrt{ \frac{2 \Lambda(\sample) \log(8/\delta)}{ \ell_1 } } +
		\frac{7 D \log(8/\delta)}{3 (\ell_1 - 1)}$;
	
	$\hat{v} \gets \hat{d}^2 \max_{v \in V} \Bigl\{ \tilde{p}(v) + \sqrt{ \frac{2  \tilde{p}(v) \log(4/\delta)}{ \ell_1 } } +
		\frac{\log(4/\delta)}{3 \ell_1}\Bigr\}$;
	
	$\hat{x} \gets \hat{d}/2 - \sqrt{\hat{d}^2/4 - \min \{ \hat{d}^2/4 ,  \hat{v} \} }$;
	
	$\ell \gets \sup_{x \in (0,\hat{x}]} \biggl\{ \frac{ \hat{d}^2 \ln \left( \frac{4 \hat{d} \hat{\rho} }{x \delta } \right) } { g(x)h \left( \frac{\varepsilon \hat{d}}{g(x) } \right) } \biggr\} $; \label{alg:firstsamplingend}
	
	$\sample \gets \text{\samplalg}(G, \{ x_v \} , \ell)$; \label{alg:secondsampling}
	
	\lForAll{$v \in V$}{$\tilde{p}(v) \gets \frac{1}{\ell} \sum_{i=1}^{\ell} \frac{\kappa(s,t,v)}{\tilde{\kappa}(s,t)} \ind{v \in \intern ( \tau^i_{st} ) }$}
	
	\textbf{return} $\{ \tilde{p}(v) , v \in V \}$
	
\end{algorithm2e}

The input of \percis\ is composed of the graph $G$, 
the percolation states $\{ x_v \}$ of all nodes, 
an integer $\ell_1 \geq 2$, 
and the approximation parameters $\varepsilon , \delta$. 
First, the algorithm computes an upper bound $D$ to the vertex diameter of $G$, using a call to the procedure \textsc{VertexDiamUB}. 
The implementation of this procedure depends on the type of considered graph. 
For instance, when $G$ is undirected and unweighted, a simple BFS from an arbitrary node of $G$ allows obtaining a constant factor approximation of the diameter~\cite{Riondato_2016}. 

The algorithm, at high level, follows two sampling phases:
in the first phase (lines~\ref{alg:firstsamplingstart}-\ref{alg:firstsamplingend}), it draws $\ell_1$ random samples from the importance distribution $q$, which are used to compute the values $\hat{\rho}$, $\hat{v}$, and $\hat{x}$ in a data-dependent manner; 
$\hat{\rho}$ is used an high probability upper bound to the sum of all percolation centralities $\sum_v p(v)$, and is related to the \emph{average} number of internal nodes of the shortest paths,  
while $\hat{v}$ is an upper bound to the maximum variance $\max_v \Vard{\tilde{p}(v)}{q}$. 
As we will prove in Section~\ref{sec:percisanalysis}, these quantities 
tightly control the number of random samples $\ell$ (that is computed in line~\ref{alg:firstsamplingend}) 
to use 
in the second phase of the algorithm (starting in line~\ref{alg:secondsampling}). 
Note that accurate estimates of $\hat{\rho}$ and $\hat{v}$ 
are obtained even when 
$\ell_1$ 
is small, as we prove in Section~\ref{sec:percisanalysis} (in practice we set $\ell_1 = \max\{10^3 , \ln(1/\delta)/\varepsilon \}$, which is always $\ell_1 \ll \ell$). 
In the second sampling phase, 
\percis\ draws $\ell$ random samples and computes the approximation $\{ \tilde{p}(v) , v \in V \}$ 
given in output. 

The most expensive operation performed by \percis\ 
regards sampling shortest paths from the importance distribution $q$;
to do so, the algorithm calls the procedure $\text{\samplalg}(G, \{ x_i \} , \ell)$. 
This procedure, described by Algorithm~\ref{algo:nonuniform_sampling_alg} 
\ifwithappendix
(in the appendix due to space constraints), 
\else
(in the appendix~\cite{percisarxiv} due to space constraints), 
\fi
takes in input $G$, 
the percolation states $\{ x_i \}$ of the nodes, 
and the number $\ell$ of random samples to generate.

\samplalg\ is based on the following idea: 
it first samples the starting node $s$ with marginal probability $\sum_{u \in V} \tilde{\kappa}(s , u)$; 
then, conditioning on the chosen start node $s$, it samples $t$ with probability 
$\tilde{\kappa}(s,t) / \sum_{u\in V}\tilde{\kappa}(s,u)$. 
Sampling a pair $s,t$ according to this process can be trivially achieved in time $\bigO(n)$ \emph{per sample}, which is too demanding in practice; 
instead, we implement both random choices with a binary search. 
After a $\bigO(n)$ preprocessing (lines~\ref{alg:preprocessingstart}-\ref{alg:preprocessingend}), 
each pair $s,t$ is sampled in $\bigO(\log n)$ time (lines~\ref{alg:samplingstart}-\ref{alg:samplingend}). 

Once a pair of nodes $s,t$ is sampled, 
the algorithm computes
the set $\Gamma_{st}$ of shortest paths between $s$ and $t$ 
and samples one shortest path from such set
using the procedure $\text{\textsc{RandomSP}}(G,s,t)$. 
This procedure can be implemented in time $\bigO(m)$ using a (truncated) BFS, which is initialized from $s$ and expanded until $t$ is found. 
However, we can significantly speed-up this task. 
Indeed, by using a \emph{balanced bidirectional BFS}~\cite{Dantzig+1963}, 
i.e., a balanced expansion of two BFS from both $s$ and $t$, this computation can be much more efficient in practice 
(e.g., completed in time $\bigO(\sqrt{m}\,)$ on several random graph models and in realistic instances~\cite{borassi2019kadabra}) while featuring remarkable theoretical properties~\cite{haeupler2025bidirectional}. 
Therefore, we use this graph traversal technique to speed-up the sampling procedure.

We summarize these observations in the following result, that provides a bound to the running time of \samplalg\ and, consequently, of \percis.

\begin{proposition}
\label{prop:runningtime}
\samplalg\ correctly draws $\ell$ samples from $q$ in time $\bigO(n+\ell( \log n + T_{\text{\textsc{BBFS}}}))$
and space $\bigO(n+m)$, 
where $T_{\text{\textsc{BBFS}}}$ is the time to run the balanced bidirectional BFS.
\end{proposition}

We remark that standard approaches, based on uniform sampling, require time $\bigO(\ell T_{\text{\textsc{BBFS}}})$ to generate $\ell$ random samples; 
this suggests that the additional $\bigO( \ell \log n)$ overhead term, due to the use of Importance Sampling, is negligible. 
We verify this observation in our experimental evaluation. 

\subsection{Analysis of \percis}
\label{sec:percisanalysis}

In this section we prove the correctness of \percis.
The main result we develop is a novel bound, stated in Theorem~\ref{thm:sample_size}, 
to the number of random samples required to obtain a high-quality approximation of the percolation centrality.
We leverage the Importance Sampling scheme described in the previous sections 
and advanced concentration inequalities~\cite{Maurer_2009,boucheron2013concentration}. 
\ifwithappendix
Due to space constraints, the proofs of our results are deferred to the appendix. 
\else
Due to space constraints, the proofs of our results are deferred to the online appendix~\cite{percisarxiv}. 
\fi

To prove Theorem~\ref{thm:sample_size}, we need to define some quantities and preliminary results. 
First, we define a key parameter $\rho$ that characterizes our bound, that is the \emph{weighted average number of internal nodes} in a shortest path, defined as
$\rho = \sum_{s,t\in V} |I(\tau_{st})| \tilde{\kappa}(s,t) $. 
A key observation is that $\rho$ is tightly related to the sum of all percolation centralities. 
\begin{lemma}
\label{lemma:sumofpercs}
It holds $ \rho \leq \sum_{v \in V} p(v) \leq \hat{d} \rho $.
\end{lemma}
Our new bound is based on the fact that, similarly to the betweenness centrality~\cite{Pellegrina_2023}, 
the percolation centrality 
satisfies a form of negative correlation among the vertices of the graph: 
since the sum of all percolation centralities is bounded by $\hat{d} \rho$, 
the existence of a node $v$ with high percolation centrality $p(v)$ constraints the centralities of the other nodes $u \neq v$ to satisfy $\sum_{u \in V , u \neq v} p(u) \leq \hat{d} \rho - p(v)$. 
This means that, intuitively, the number of vertices with high percolation centrality cannot be arbitrarily large, and, consequently, the number of nodes with large approximation errors should roughly depend on $\rho$ rather than $n$ (the number of nodes of the graph). 
On the other hand, 
the state-of-the-art bound for the betweenness centrality 
(Theorem~4.7 in~\cite{Pellegrina_2023}) 
is based on the 
\emph{unweighted} average number of internal nodes,
and applies to the simpler 
\emph{uniform} random sampling distribution, 
thus cannot be adapted directly to our setting, i.e., for 
percolation centrality and an Importance Sampling estimator.
We address these issues with Theorem~\ref{thm:sample_size}.
Let $g(x)=x(\hat{d}-x)$ and $h(x)=(1+x)\ln(1+x)-x$. 
\begin{theorem}\label{thm:sample_size}
Define $\hat{v}$, $\hat{x}$, and $\hat{\rho}$
 such that
 \begin{align*}
     & \max_{v\in V} \Vard{ \tilde{p}(v) }{q}\leq \hat{v} \leq \hat{d}^2/4, \text{  } \sum_{v\in V}p(v)\leq \hat{d} \hat{\rho} ,  \\ 
     & \text{ and } \hat{x} = \hat{d}/2 - \sqrt{\hat{d}^2/4 - \hat{v} } .
 \end{align*}
    For $\delta,\varepsilon\in (0,1)$, let $\mathcal{S} \! = \! \{ \tau^1, \dots, \tau^\ell \}$ be a sample of $\ell$ shortest paths drawn i.i.d. from the importance distribution $q$, with 
    \begin{align*}
    \ell = \sup_{x \in (0,\hat{x}]}\left\{ \frac{ \hat{d}^2 \ln \left( \frac{2 \hat{d} \hat{\rho} }{x \delta } \right) } { g(x)h \left( \frac{\varepsilon \hat{d}}{g(x) } \right) } \right\} .
     \end{align*}
With probability $\geq 1-\delta$ over $\calS$, 
   $ |\tilde{p}(v) - p(v)| \leq \varepsilon , \forall v \in V $.
 \end{theorem}
 We remark that, for typical values of the parameters $\varepsilon$ and $\delta$, the required number of samples $\ell$ is approximately 
$\frac{ \left( 2\hat{v}+\frac{2}{3}\varepsilon \hat{d} \right)}{\varepsilon^2}\left(\ln ({\hat{d} \hat{\rho} }/{\hat{v}} )+\ln\left(2/\delta\right)\right) $. 
Interestingly, this bound is smaller than the bound $\mathcal{O}(\ln(D/\delta)/\varepsilon^2)$ provided by Lima et al.~\cite{Lima_2020}, 
even if our guarantees are \emph{much stronger}, since they allow obtaining tight approximations to the percolation centralities $\{ p(v) , v \in V \}$ (instead of the doubly normalized variant $\{ p^*(v) , v \in V \}$). 
In fact, on real world graphs, 
$\hat{\rho} \ll D$ and the maximum variance is typically much smaller than its trivial upper bound, i.e., $\hat{v} \ll \hat{d}^2/4 $. 

We now provide a sharp, high-confidence upper bound $\hat{\rho}$ to the weighted average number of internal nodes $\rho$, which is a key quantity in our sample complexity analysis, and is computed in the first sampling phase of \percis. 
By applying an Empirical Bernstein bound~\cite{Maurer_2009}, we obtain a data-dependent estimate that adapts to the graph structure. 

\begin{proposition}
\label{prop:rhohatbound}
Let $D \geq \max_\tau | \intern(\tau) |$ be an upper bound to the vertex diameter of $G$, 
and $\mathcal{S}$ be 
a sample $\mathcal{S} = \{ \tau^1, \dots, \tau^\ell \}$ of $\ell$ shortest paths drawn i.i.d. from the importance distribution~$q$. 
Let $\tilde{\rho}(\mathcal{S}) = \frac{1}{\ell}\sum_{i=1}^\ell | \intern(\tau_i) | $. 
Then, define $\Lambda(\mathcal{S})$ and $\hat{\rho}$ as
	\begin{align*}
		\Lambda(\mathcal{S}) & = \frac{1}{\ell(\ell-1)} \sum_{1 \le i < j \le \ell} \left(
		| \intern(\tau_i)| - |\intern(\tau_j)|
		\right)^2 , \\
		\hat{\rho} & = \tilde{\rho}(\mathcal{S}) +
		\sqrt{ \frac{2 \Lambda (\sample) \log(2/\delta)}{\ell} } +
		\frac{7 D \log(2/\delta)}{3(\ell - 1)}.
	\end{align*}
	Then, with probability $\geq 1 - \delta$ it holds $\sum_{v \in V} p(v) \leq  \hat{d} \hat{\rho}$. 
\end{proposition}
We now prove a data-dependent bound $\hat{v}$ to the maximum estimation variance $\max_v \Vard{\tilde{p}(v)}{q}$. 
\begin{proposition}
\label{prop:boundmaxvar}
Let $\mathcal{S}$ be 
a sample $\mathcal{S} = \{ \tau^1, \dots, \tau^\ell \}$ of $\ell$ shortest paths drawn i.i.d. from the importance distribution $q$. 
Define $\hat{v}$ as
	\begin{align*}
	\hat{v} = \hat{d}^2 \max_{v \in V} \biggl\{ \tilde{p}(v) + \sqrt{ \frac{2 \tilde{p}(v) \log(1/\delta)}{ \ell } } +
		\frac{ \log(1/\delta)}{3 \ell} \biggr\} .
	\end{align*}
	Then, with probability $\geq 1 - \delta$ it holds $\max_v \Vard{\tilde{p}(v)}{q} \leq  \hat{v}$.
\end{proposition}
The correctness of \percis\ follows by combining all the results proved in this section.
\begin{proposition}
\label{prop:algcorrect}
The output of \percis\ is an $\varepsilon$-approximation of $\{ p(v) , v \in V\}$ with probability $\geq 1-\delta$.
\end{proposition}

\subsection{Comparison of \percis\ with \unif}
\label{sec:percisanalysiscomparison}

In this section we 
further analyze \percis, 
and we compare it with the standard approach based on drawing samples  according to the uniform distribution, as done by previous works~\cite{Lima_2020,Lima_2022}. 
We note that all our results hold for both algorithms presented in~\cite{Lima_2020,Lima_2022} since they share the same sampling distribution. 
To ease the presentation, we denote such algorithm as \unif. 

First, we prove that, under mild assumptions on the values of the percolation states, the likelihood ratio $\hat{d}$ of our Importance Sampling scheme is bounded by a constant. 
This implies that our sampling scheme provably achieves high-quality approximations in an efficient manner. 
First, we define 
$\Delta = \min_v \max_{s \neq v \neq t}( x_s - x_t) $.
The parameter $\Delta$ implies the following simple fact: 
for every node $v \in V$, 
there exist two nodes $s,t$ with $s \neq v \neq t$, such that $x_s \geq x_t + \Delta$.  
We assume that $\Delta$ is not too small, i.e., that is a constant $\Delta \in \Omega(1)$. 
For instance, consider a graph of $\geq 3$ nodes, with percolation states equal to $1$, $1/2$, and $0$;
in this case, $\Delta = 1/2$.
Note that this assumption is extremely mild, and satisfied by all reasonable instances (for all the networks we considered, it holds $\Delta = 1$). 
Under this setting, we prove the following bound to the likelihood ratio $\hat{d}$ for the importance sampling distribution $q$ used by \percis. 

\begin{proposition}
\label{thm:likelihoodpercis}
When $\Delta \in \Omega(1)$, 
the likelihood ratio $\hat{d}$ of the importance sampling distribution $q$ is $\hat{d} \in \mathcal{O}(1)$. 
\end{proposition}

The consequence of this result is that the approximations computed by PercIS of the percolation centralities of all nodes are guaranteed to converge rapidly to their expected values, thanks to the guarantees from Theorem~\ref{thm:sample_size}. 
In strong contrast with this positive result, we show that there exist instances (satisfying the same assumption $\Delta \in \Omega(1)$ described above) where the uniform sampling distribution yields a very large likelihood ratio $\Omega(n)$. 
Consequently, there exist instances where the returned approximations by \unif\ are expected to be poor, due to high estimator variance. 

\begin{proposition}
\label{thm:likelihoodunif}
There exist instances with $\Delta \in \Omega(1)$ where the likelihood ratio of the uniform sampling distribution is $\Omega(n)$. 
\end{proposition}

We now further highlight the gap between \percis\ and \unif. 
We prove a \emph{lower bound} to the number of samples needed by \unif: we show that there exist instances where the number of required samples for the uniform distribution, thus the running time of the algorithm, is comparable to the cost of computing the centrality values exactly. In contrast, for the same instances, \percis\ is still efficient since the \emph{upper bound} to the number of samples is linear in the graph size. 
This implies that, in such cases, \percis\ is provably more efficient than the standard approach by a factor at least $n$.

\begin{proposition}
\label{thm:uniflowerbound}
There exist instances with $\Delta \in \Omega(1)$ where at least $\Omega(n^2)$ random samples are needed by \unif, 
while $\mathcal{O}(n)$ random samples are sufficient for \percis. 
\end{proposition}

\section{Experimental Evaluation}\label{sec:experimental}
In this section, we summarize the results of our experimental study on approximating
the percolation centrality in real-world networks. 
Our main goal is to compare \percis\ with \unif, 
and to apply \percis\ to analyze real-world labelled networks.
\subsection{Networks and Experimental Setting}\label{sec:experimental_setting}
We evaluate all algorithms on graphs from several domains. 
We mainly test \percis\ and \unif\ on large real-world graphs from 
SNAP ({\url{http://snap.stanford.edu/data}}). 
We report their statistics in Table~\ref{tab:datasets} (in the appendix). 
Since these graphs are unlabelled, and the percolation states of the nodes are not available, we consider four different settings inspired to practical scenarios. 
Doing so, we robustly test the algorithms over a wide choice of states' distributions:
\begin{description}
	\item[Random Seeds (RS):] We select a constant number of nodes uniformly at random and assign them percolation state equal to $1$, while all other nodes have percolation state $0$. This experiment simulates scenarios in which there are initial seeds that contracted an infection, or equivalently, users in posses of a piece information to distribute over the network. The goal is to identify important nodes for the percolation over the network at its very first stages. The choice of the number of seeds (in our case, $50$) is similar to parameters used by previous works  (e.g.,~\cite{kempe2003maximizing}). 
	\item[Random Seeds Spread (RSS):] We select a set $K$ of $\log n$ random seeds, assign them percolation state $1$ and execute a BFS from each seed $s \in K$. To every node $v$ we assign the percolation state $x_v = \max_{s \in K}\{1/4^{d(s,v)}\}$, where $d(s,v)$ is the shortest path distance from $s$ to $v$. 
This experiment models the risk of a diffusion from $K$, setting the states of the nodes to the chance that they are reached (assuming each node spreads to his neighbors with probability $1/4$);
in such a setting, central nodes could be relevant for the percolation over the network at an intermediate diffusion state.
\item[Isolated Component (IC):] Given the input graph $G$, we add a (strongly) connected path $P$ of constant length ($50$ nodes). 
When $G$ is directed, we include an edge from a random node of the largest connected component of $G$ 
with the first node of $P$; otherwise, if $G$ is undirected, $P$ forms an isolated connected component.  
	 We set the percolation states of half of the nodes in $P$ (i.e., $25$ nodes) to $1$ and the rest to $0$. This setting is similar to the ``worst-case'' instance used in the proof of Proposition~\ref{thm:uniflowerbound}, so it is useful to empirically verify the theoretical gap between \percis\ and \unif. 
	\item[Uniform States (UN):] We assign a uniform random value in $[0,1]$ to each percolation state, as done by~\cite{Lima_2020,Lima_2022}. 
\end{description}
Then, as a case study and application of \percis\ to large Labelled Network (LN) analysis, we considered the labeled social networks in~\cite{garimella2018political,Preti_2025} 
	and video recommendation networks from~\cite{ribeiro2020auditing,mamie2021anti,coupette2023reducing} (see Table~\ref{tab:datasets_labeled} in the appendix for the statistics). 
	For social networks, each node has a categorical attribute encoding its opinion w.r.t. a polarizing topic (e.g., pro vs. against gun control); we selected graphs on which there are only two opinions and interpreted them as percolation states $x_v \in \{0,1\}$. 
	In this application, central nodes are the ones that connect users of opposing views, i.e., acting as bridges between the two polarized communities. 
	For recommendation networks, we consider the Youtube graph, where a node is a video and edges are to-watch-next recommendations. Each node $v$ has an attribute $r_v$ (a real-valued score $r_v \in [0,1]$) representing the content's level of radicalization. 
	For this application, we assign to each $v$ the percolation state $x_v = 1-r_v$ (i.e., lower state to more harmful videos). 
	Under this setting, central nodes 
	may be important for radicalization pathways, i.e., they allow reaching harmful content from safe nodes. 

We implemented all the algorithms in \texttt{C++}.
The code is available at \url{https://github.com/Antonio-Cruciani/PERCIS}.
All experiments have been executed on a server running Rocky Linux 8.10 equipped with AMD Epic 7413 processor for overall 30 cores and 498 GB of RAM. For the exact values, we implemented the algorithm 
in~\cite{Chandramouli_2021}, 
which is an extension of Brandes' algorithm~\cite{brandes2001faster} to percolation centrality. 
We fix $\delta = 0.05$, 
and report averages and stds over 10 runs.

\subsection{Experimental Results}
\paragraph{Guarantees of \unif} 
As a first experiment, we 
empirically verify the theoretical guarantees of \unif, 
the approximation method of~\cite{Lima_2020,Lima_2022} 
based on the uniform sampling distribution; 
our goal is to test whether \unif\ actually computes accurate approximations of the percolation centralities $\{ p(v) , v \in V \}$. 
Recall that, as discussed in Section~\ref{sec:percapproxprelims}, 
\unif\ provides guarantees w.r.t. the doubly normalized percolation score $p^*(v) = \frac{p(v)}{n(n-1)}$, 
but not necessarily on the measure $p(v)$ given in Definition~\ref{def:percolationcentr}. 
For this experiment we vary the accuracy parameter $\varepsilon \in \{ 0.05,0.01,0.005,0.001,0.0005 \}$; 
then, for each $\varepsilon$, we compute the number of samples $\ell$ according to the bound $\mathcal{O}(\ln(D / \delta) / \varepsilon^2)$ of~\cite{Lima_2020} (that is also an upper bound to the number of samples for the algorithm in~\cite{Lima_2022}); 
we draw $\ell$ samples from the uniform distribution, 
we compute the estimates $\{ \tilde{p}^*(v) , v \in V \}$, 
and then evaluate the 
Maximum (absolute) Error (ME) $\xi^{\text{\textsc{U}}} = \max_v | n(n-1)\tilde{p}^*(v) - p(v) |$. 
Note that we achieve an $\varepsilon$-approximation 
when it holds $\xi^{\text{\textsc{U}}} \leq \varepsilon$. 
Figure~\ref{fig:maxerrorsmain} shows the results for the RS, RSS, and UN experimental settings. 
While for the UN setting the MEs are usually below $\varepsilon$, 
for the other settings the errors are significantly larger than $\varepsilon$ (points above the diagonal). 
This experiment confirms that, in general, a 
sample that yields a good approximation of 
$\{ p^*(v) , v \in V \}$
is not necessarily valid for an accurate approximation of the percolation scores 
$\{ p(v) , v \in V \}$. 
Therefore, we clearly conclude that \unif\ does not provide sound theoretical guarantees for the problem of approximating the percolation centralities $\{ p(v) , v \in V \}$.

\begin{figure*}[htb!]
\centering
	\includegraphics[width=\linewidth]{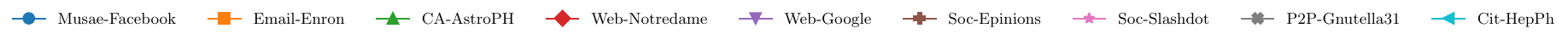}
	\begin{subfigure}{0.22\textwidth}
		\includegraphics[width=\linewidth]{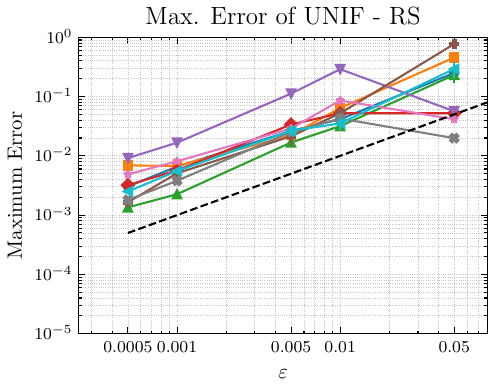}
		\caption{}\label{fig:sd_pseudo_ri}
	\end{subfigure}
	\begin{subfigure}{0.22\textwidth}
		\includegraphics[width=\linewidth]{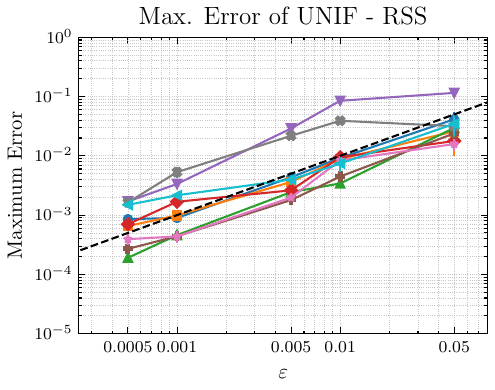}
		\caption{}\label{fig:sd_pseudo_is}
	\end{subfigure}
	\begin{subfigure}{0.22\textwidth}
		\includegraphics[width=\linewidth]{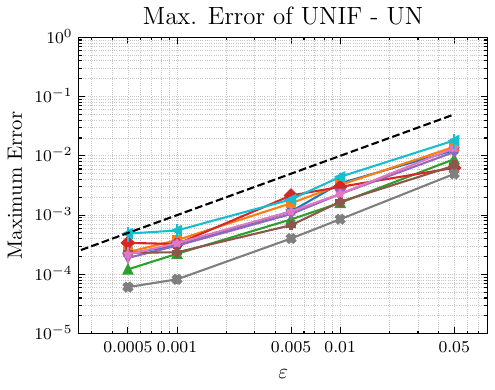}
		\caption{}\label{fig:sd_pseudo_un}
	\end{subfigure}
	\begin{subfigure}{0.22\textwidth}
		\includegraphics[width=\linewidth]{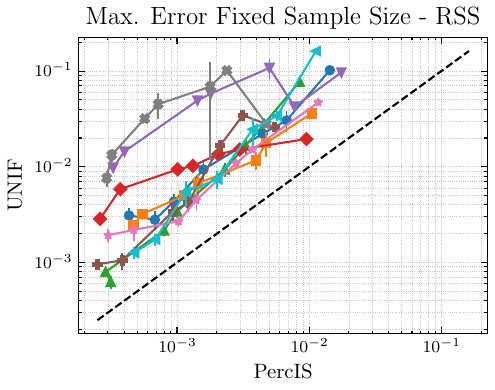}
		\caption{}\label{fig:sd_spread}
	\end{subfigure}
	\caption{(a-c): Maximum Error of \unif\ on random samples of size $\mathcal{O}(\ln(D/\delta)/\varepsilon^2)$ (the bound of~\cite{Lima_2020}) for $\varepsilon\in[0.05,0.0005]$.
	(d): Maximum Errors of \percis\ ($x$ axes) and \unif\ ($y$ axes) on random samples of fixed sizes $\ell \in [10^3, 10^6 ] $ for RSS (other settings shown in Figure~\ref{fig:sd_fixed_ss_appendix} in the appendix).
	}\label{fig:maxerrorsmain}
	\vspace{-16px}
\end{figure*}

\paragraph{Maximum Error on fixed sample sizes}\label{sec:max_error} 
We now compare the ME $\xi^{\text{\textsc{U}}}$ obtained by 
\unif\ with the 
ME $\xi^{\text{\textsc{P}}} = \max_v | \tilde{p}(v) - p(v) |$ 
of \percis\
using the same number $\ell$ of random samples. 
In this way, we directly asses the estimation accuracy of our novel 
importance distribution, introduced in Section~\ref{sec:importancesampling}. 
For this comparison, we fix the sample size $\ell$ for both distributions 
to the values 
$\ell\in \{10^3,5\cdot 10^3,10^4,5\cdot 10^4,10^5,5\cdot 10^5,10^6\}$;
for each $\ell$, we draw $\ell$ random samples and compute the MEs~$\xi^{\text{\textsc{U}}}$ and $\xi^{\text{\textsc{P}}}$. 

Figure~\ref{fig:maxerrorsmain} shows the results for these experiments for the RSS setting (other settings are similar and shown in Figure~\ref{fig:sd_fixed_ss_appendix} in the appendix due to space constraints). 
Note that each plot shows the MEs of \unif\ ($y$ axes) and \percis\ ($x$ axes), where each point of the plot corresponds to a value of $\ell$, while the diagonal black line is at $y=x$. 
From these figures we observe that the importance distribution of \percis\ consistently and significantly outperforms the uniform distribution on \emph{every} graph and \emph{every} setting, 
with an improvement on the ME of up to two orders of magnitude. 
Figure~\ref{fig:sd_component} (in the appendix) shows more detailed results for the IC setting for the largest graph (Web-Google); here the comparison is done between the exact values ($y$ axes) and the random estimates ($x$ axes). 
In this figure, it is clear that, even by using a very large number of random samples ($\ell = 10^6$), \unif\ does not provide any reasonable result (as it reports all estimates equal to $0$), while \percis\ reports all values very close to the exact scores. 

These findings clearly show that the importance distribution used by \percis\ is vastly superior in terms of accuracy of the estimates. This is a consequence of the refined theoretical properties of the Importance Sampling distribution proved in previous sections, and 
the large variance that characterizes the uniform distribution on natural realistic instances.
Such issues arise from ignoring the percolation states of the nodes when drawing random shortest paths.

\paragraph{Target Maximum Error}
In this third experiment, we compare \percis\ with \unif\ by measuring the number of samples required to approximate percolation centrality within a prescribed error bound $\varepsilon$. More precisely, we run both approximation algorithms at increasing sample sizes until the ME is $\leq \varepsilon$.
To ensure a meaningful $\varepsilon$-approximation, we set $\varepsilon$ to be at most the maximum exact percolation centrality scores. Therefore, for each graph and setting, we set $\varepsilon = (1/k)\cdot \max_{v\in V}p(v)$ for $k\in\{2,4,5,10\}$. 
To measure the required samples we use the following procedure: we start with $\ell =10^3$ samples, compute the ME $\xi$, and check whether $\xi\leq \varepsilon$. If not, we increment the sample size by $10^3$, and repeat. The experiment terminates either when $\xi\leq \varepsilon$ or when the total number of samples reaches $10^8$.

We report the results in Figure~\ref{fig:target_ss_sd}. The plots show the sample sizes needed by \percis\ ($x$-axes) and \unif\ ($y$-axes) to reach an ME of at most $\varepsilon$. 
We clearly observe that the importance distribution of \percis\ consistently achieves $\varepsilon$-approximations using significantly fewer samples across all experimental settings. In contrast, \unif\ generally requires sample sizes of several orders of magnitude larger compared to \percis, especially in the RS and IC settings.
Often, \unif\ hits the $10^8$ cap without achieving the target $\varepsilon$-approximation. 
In the IC setting, \unif\ fails entirely: for many instances it always reaches $10^8$ samples with $\xi^{\text{\textsc{U}}} \gg \varepsilon$,
while \percis\ always converges after at most $2 \cdot 10^6$ samples. This observation empirically confirms the theoretical lower bound for \unif\ proved by Proposition~\ref{thm:uniflowerbound} and the corresponding strong gap w.r.t. \percis. 
Moreover, 
we remark that the number of random samples at which the procedure stops is 
a \emph{lower bound} to the number of samples that any approach would need to guarantee an $\varepsilon$-approximation; 
the observed behavior rules out any improvement that may be achieved with more advanced techniques (e.g., Rademacher Averages~\cite{Pellegrina_2023}) when used with the uniform distribution. 
In strong contrast, the importance distribution used by \percis\ is much more accurate and always converges at practical sample sizes.

\begin{figure*}[htb!]
\centering
	\includegraphics[width=\linewidth]{./img/experiments/shared_legend}
	\begin{subfigure}{0.22\textwidth}
		\includegraphics[width=\linewidth]{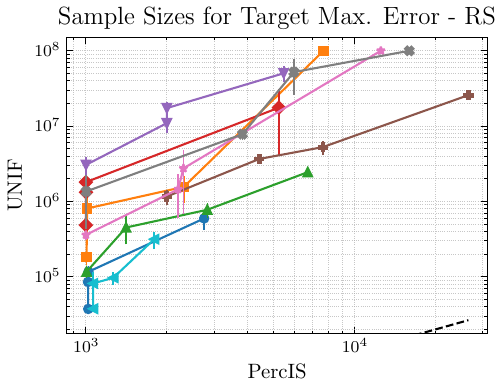}
	\end{subfigure}
	\begin{subfigure}{0.22\textwidth}
		\includegraphics[width=\linewidth]{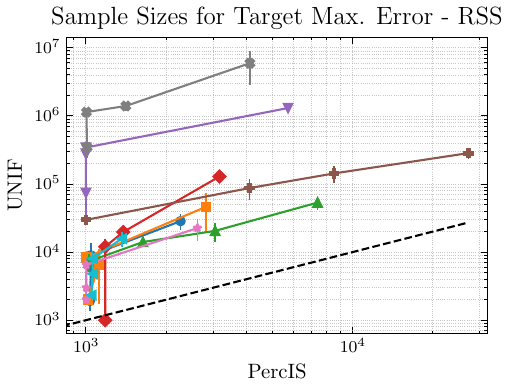}
	\end{subfigure}
	\begin{subfigure}{0.22\textwidth}
		\includegraphics[width=\linewidth]{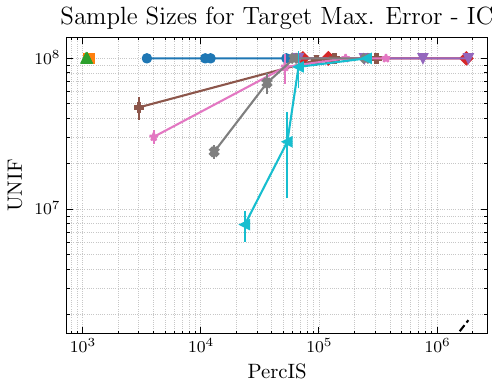}
	\end{subfigure}
	\begin{subfigure}{0.22\textwidth}
		\includegraphics[width=\linewidth]{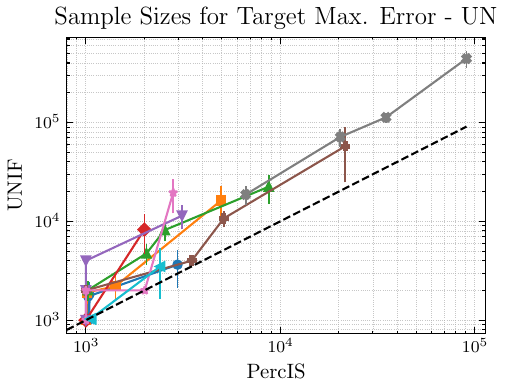}
	\end{subfigure}
	\caption{
	Sample sizes required to obtain a Maximum Error $\leq \varepsilon$ by
	\unif\ ($y$ axes) and \percis\ ($x$ axes). 
\vspace{-10px}
	}
	\label{fig:target_ss_sd}
\end{figure*}

\begin{figure*}[htb!]
\centering
	\begin{subfigure}{0.22\textwidth}
		\includegraphics[width=\linewidth]{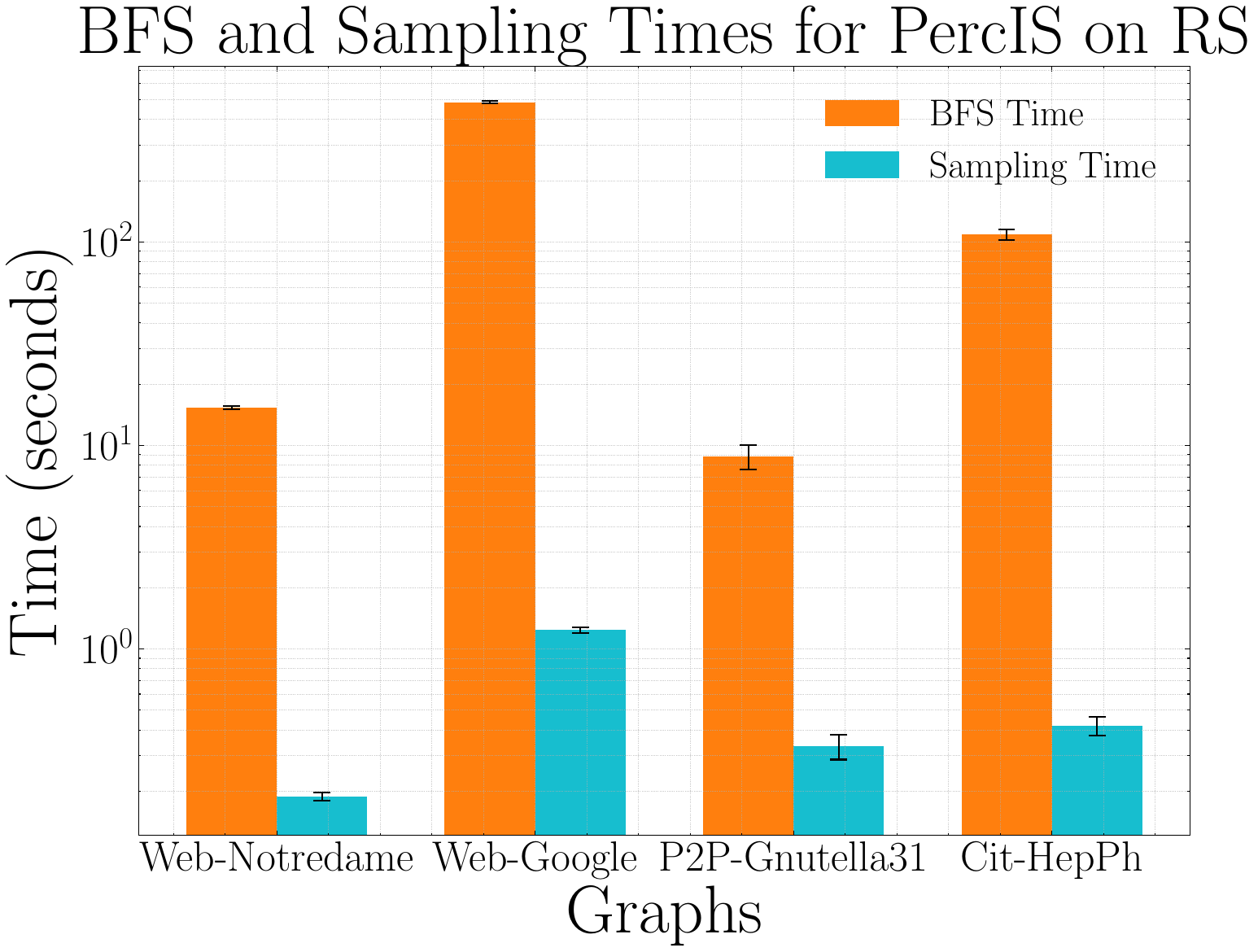}
		\caption{}\label{fig:bfs_sampling_time_rnd_init}
	\end{subfigure}
		\begin{subfigure}{0.22\textwidth}
		\includegraphics[width=\linewidth]{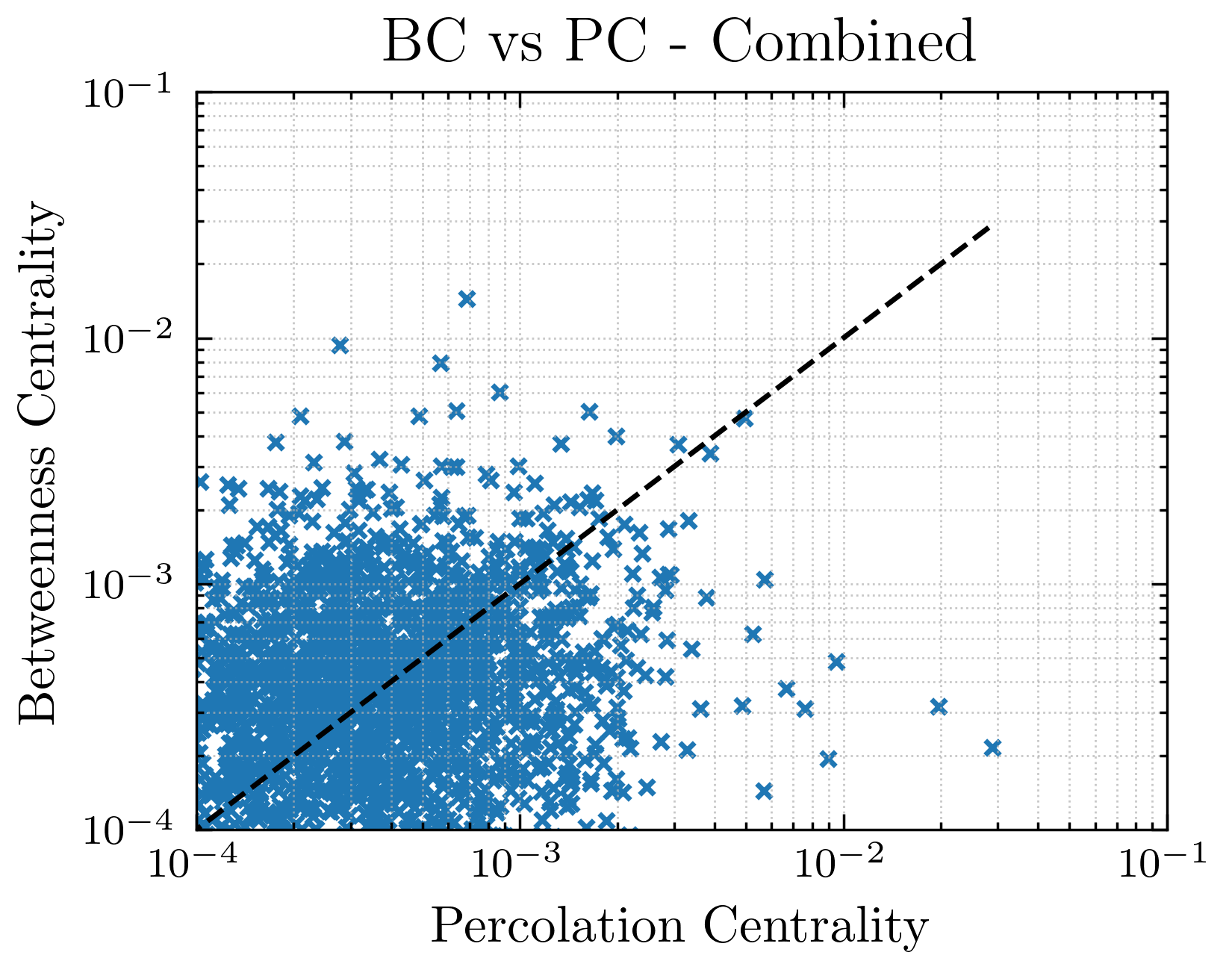}
		\caption{}\label{fig:pc_vs_bc_comb}
	\end{subfigure}
	\begin{subfigure}{0.22\textwidth}
		\includegraphics[width=\linewidth]{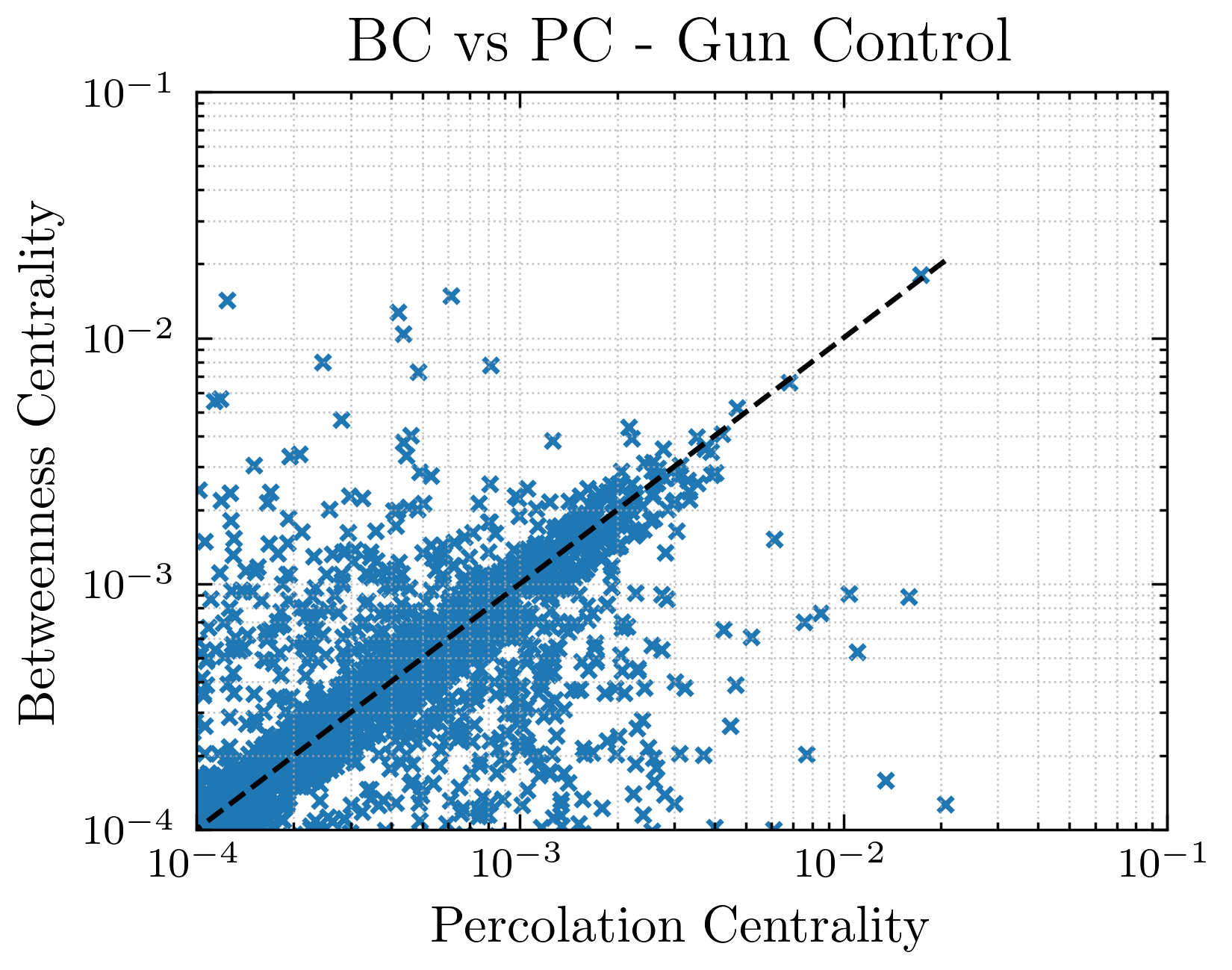}
		\caption{}\label{fig:pc_vs_bc_gun}
	\end{subfigure}
	\begin{subfigure}{0.22\textwidth}
		\includegraphics[width=\linewidth]{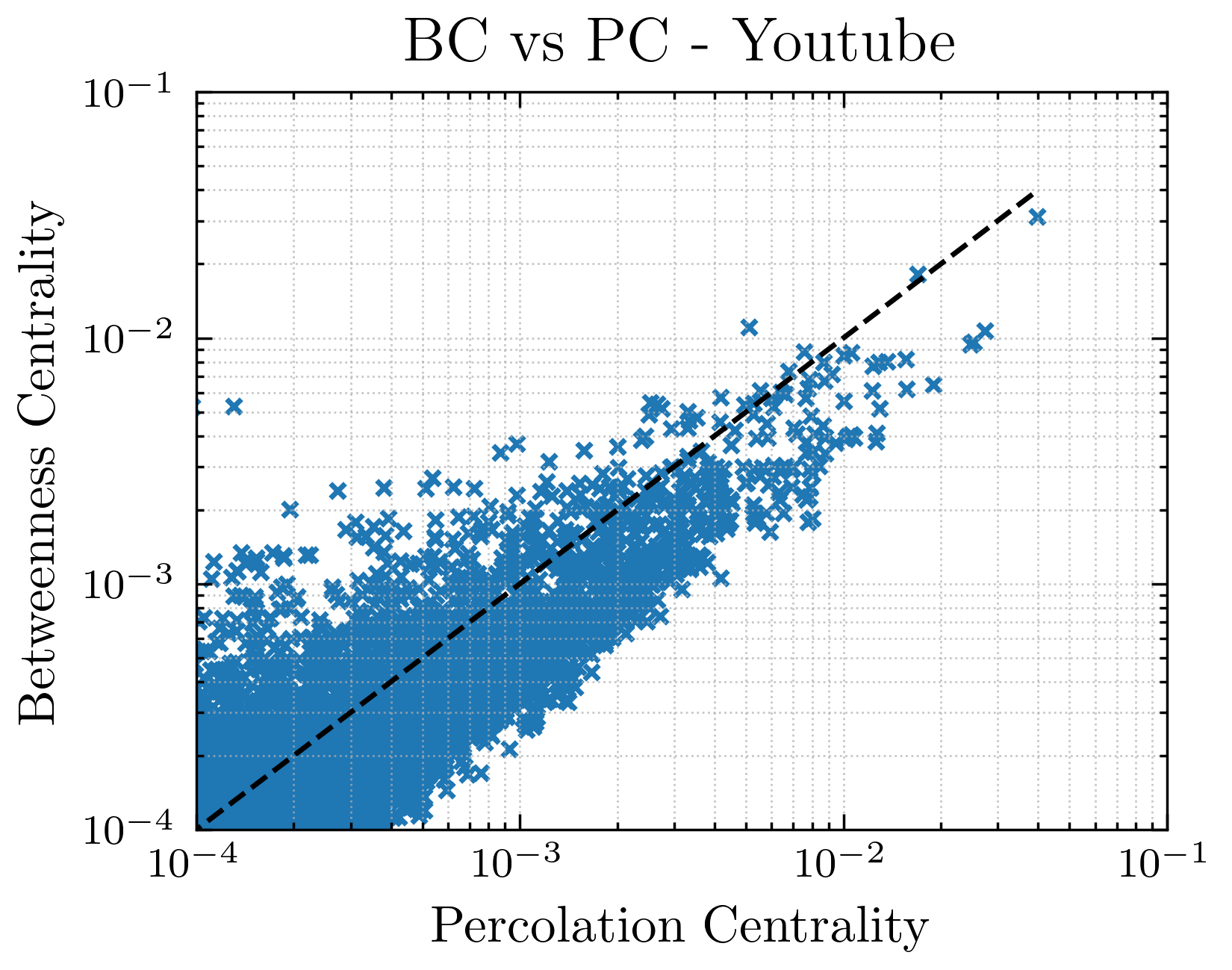}
		\caption{}\label{fig:pc_vs_bc_youtube}
	\end{subfigure}
	\caption{(a) Overall time needed by the bidirectional BFSs (orange) and Importance Sampling (blue), for $\ell = 10^6$ samples.
	(b-d): Percolation Centrality and Betweenness Centrality scores of the most central nodes from Labeled Networks.
	\vspace{-12px}
	}\label{fig:bfs_sampling_time_realworld}
\end{figure*}

\paragraph{Running Times}
In this experiment we evaluate the running time of \percis\ and compare it with \unif\ to verify the overhead of the Importance Sampling scheme. 
As before, we run the algorithms with a fixed number of samples $\ell\in [ 10^3, 10^6 ] $ and measure the running times. 
We remark that both methods use the same algorithm to traverse the graph (the bidirectional BFS), so any difference between the two methods depends on the sampling distribution. 
In Figure~\ref{fig:time_settings} (in the appendix), we compare the running times of \percis\ with \unif, 
observing that the two methods have comparable running times. 
In Figure~\ref{fig:bfs_sampling_time_realworld} we compare the total time needed by \percis\ for executing the BFSs and for sampling pairs of nodes $s,t$ from the importance distribution, 
on $10^6$ samples and RS (other plots are similar, shown in Fig.~\ref{fig:bfs_sampling_time_appendix} in the appendix). 
We note that the sampling time is orders of magnitude smaller than the BFSs.
These results confirm that 
the overhead of Importance Sampling is negligible.

Moreover, \percis\ is significantly more efficient than the exact algorithm. On the two largest graphs (Web-Notredame and Web-Google) and the most demanding setting (UN states, $10^6$ samples), \percis\ completes in $6.5$ and $18.3$ seconds respectively. In contrast, the exact algorithm requires $3203$ and $86914$ seconds, yielding speedups of approximately $\times 492$ and $\times 4829$. These results highlight not only the efficiency of \percis\ but also the impracticality of running the exact algorithm in dynamic settings (e.g., where the graph or the percolation states may evolve over time due to the spread of an infection) or to handle larger instances, typical of  applications.

\paragraph{\percis\ sample size bound}
In this experiment we evaluate the impact to the sample size $\ell$ of the data-dependent estimates $\hat{\rho}$ and $\hat{v}$ computed in the first phase of Algorithm~\ref{algo:mainalg}. 
To this end, 
we vary the accuracy parameter 
$\varepsilon\in \{0.05,0.01,0.005,0.001,0.0005\}$
and compute the bound $\ell$ using \percis\ and a variant of \percis: 
this approach skips the first phase, 
and uses much simpler values for the parameters $\hat{\rho}$ and $\hat{v}$.
More precisely, 
it sets $\hat{\rho} = D$ (where $D$ is an approximation of the vertex diameter) and $\hat{v} = \hat{d}^2/4$, i.e., its maximum possible value. 
We denote this variant as \percis-DI, since it uses a \emph{data independent} bound and does not take into account key properties of the input instance. 
We observe that, for \percis-DI, $\ell = \mathcal{O}(\ln(D/\delta)/\varepsilon^2)$, which is similar to the one proved by previous works~\cite{Riondato_2016,Lima_2020}. 
Furthermore, the likelihood ratio $\hat{d}$ of \percis\ is small and close to $1$ in all cases (e.g., it was always $\hat{d} \leq 1.05$). 

In Figure~\ref{fig:ss_vc_vs_rho} we compare the sample sizes bounds $\ell$ of 
\percis\ and \percis-DI for all experimental settings. 
In general, the data-dependent upper bound is always sensibly smaller than the DI one, up to $3$ orders of magnitude (for the IC setting, Figure~\ref{fig:ss_comp}). 
We conclude that our data-dependent scheme is very effective, and provides significantly tighter bounds w.r.t. 
previous works.

\paragraph{Application of \percis\ to labeled networks}
In this final experiment, we apply \percis\ to Labeled Networks (LN) from online social media and video recommendations. 
In Figure~\ref{fig:bfs_sampling_time_realworld} we compare the betweenness and percolation scores using an $\varepsilon$-approximation of both measures ($\varepsilon = 10^{-4}$) on LN. The results reveal a clear divergence between the two measures, which also impacts the \emph{ranking} of the most central nodes.
Unlike betweenness, percolation centrality is sensible to the percolation states of the nodes, thus it has more potential to identify nodes that play a crucial role in information spread or contagion propagation in a network.
In the applications we considered, percolation centrality may be more effective in identifying nodes that connect users with opposing views, or in flagging content that form radicalization pathways.

\section{Conclusion}
In this work we presented \percis, a new algorithm for approximating the percolation centrality of all nodes of a graph. 
\percis\ is based on a new Importance Sampling distribution, which is sensible to the percolation states of the nodes, 
and an efficient sampling scheme that
incurs in a negligible overhead compared to standard methods. 
Our analysis features new sample complexity bounds, and highlights key limitations of previous techniques, which provide much looser guarantees. 
We showed that uniform sampling approaches are not guaranteed to obtain high-quality approximations efficiently, as we proved strong lower bounds to their performance. 
We tested \percis\ on large real-world networks, under several experimental settings,
observing that it consistently outperforms the state-of-the-art in terms of accuracy and required resources.
\percis\ enables computing high-quality approximations of the percolation centrality on attributed networks, offering a new perspective on the role of the most central nodes. 

For future works, \percis\ can be extended to analyze even richer graphs, such as dynamic, uncertain, and temporal networks, all settings in which our contributions may be useful to design efficient approximation algorithms. 
Then, we believe the use of Importance Sampling in other Data Mining problems to be an interesting future direction. 

\section*{Acknowledgments}
This work was supported by 
the Research Council of Finland, Grant 363558, 
and by the Italian Ministry of University and Research~(MUR),
projects 
``National Center for HPC, Big
Data, and 
Quantum Computing" CN00000013, 
and PRIN
``EXPAND: scalable algorithms for EXPloratory
Analyses of heterogeneous and dynamic Networked Data".

\bibliographystyle{IEEEtran}
\bibliography{IEEEabrv,bibliography}
\ifwithappendix

\appendix
\appendices

\subsection{Missing Proofs}
In this section we provide the proofs to the results that could not fit in the main text due to size constraints.

\begin{proof}[Proof of Lemma~\ref{lemma:unbiasedest}]
It is immediate to observe that $\kappa(s,t,v) > 0$ implies $\tilde{\kappa}(s,t) > 0$; 
therefore, $\E_q[\tilde{p}(v)]$ is well defined. 
Then, it holds
\begin{align*}
\E_q[\tilde{p}(v)] 
&= \E_{\tau_{st} \sim q} \left[ \frac{\kappa(s,t,v)}{\tilde{\kappa}(s,t)} \ind{v \in \intern(\tau_{st})} \right] \\
&= \sum_{s,t \in V} \tilde{\kappa}(s,t) \sum_{\tau_{st} \in \Gamma_{st}} \frac{1}{\sigma_{st}} \frac{\kappa(s,t,v)}{\tilde{\kappa}(s,t)} \ind{v \in \intern(\tau_{st})} \\
&= \sum_{s,t \in V} \frac{\sigma_{st}(v)}{\sigma_{st}} \kappa(s,t,v)
= p(v).
\end{align*}
\end{proof}

\begin{proof}[Proof of Lemma~\ref{lemma:variancebound}]
First, note that from Lemma~\ref{lemma:unbiasedest} it holds $\E_q[\tilde{p}(v)] = p(v)$. 
Therefore, from the definition of variance, 
\begin{align*}
& \Vard{\tilde{p}(v)}{q} 
= \E_{\tau_{st} \sim q }\left[ \left( \frac{\kappa(s,t,v)}{\tilde{\kappa}(s,t)} \ind{v \in \intern(\tau_{st})} \right)^2 \right] - p(v)^2 \\
& =  \sum_{s,t \in V} \tilde{\kappa}(s,t) \sum_{\tau_{st} \in \Gamma_{st}} \frac{1}{\sigma_{st}} \left( \frac{\kappa(s,t,v)}{\tilde{\kappa}(s,t)} \ind{v \in \intern(\tau_{st})} \right)^2  - p(v)^2 \\
& =  \sum_{s,t \in V} \tilde{\kappa}(s,t) \frac{\sigma_{st}(v)}{\sigma_{st}} \left( \frac{\kappa(s,t,v)}{\tilde{\kappa}(s,t)}  \right)^2  - p(v)^2 \\
& \leq \hat{d} p(v)  - p(v)^2  = p(v) ( \hat{d} - p(v) ) \leq \hat{d} p(v) .
\end{align*}
\end{proof}

\begin{algorithm2e}[htb!]
	\caption{\samplalg}\label{algo:nonuniform_sampling_alg}
	\KwIn{Graph $G$, percolation states $x_1 , x_2 , \dots , x_n$ in non-increasing order, $\ell \geq 1$. }
	\KwOut{Random sample of $\ell$ shortest paths each sampled i.i.d. from the importance distribution $q$}
	
	\tcp{Preprocessing}
	
	$w_{n+1} \gets 0$; \label{alg:preprocessingstart}
	$r_{n+1} \gets 0$;
	$c \gets 0$

	\For{$i =n \textbf{ down to }1$}{
		
		$w_i \gets w_{i+1} + x_i$
		
		$c_i \gets (n-i+1) x_i - w_i$
		
		$r_{i} \gets r_{i+1} + c_i$
		
		$c \gets c + c_i$ \label{alg:preprocessingend}
		
	}

	\tcp{Sampling}
	
	$\sample \gets \emptyset$;
	
	\For{$i =1 \textbf{ to }\ell$}{ \label{alg:samplingstart}
		
		$a \gets 1$; 
		$b \gets n$; 
		$d \gets \lfloor (a+b)/2 \rfloor$
		
		$u \gets U(0,1)$
		
		\While{$a \leq b$}{
			$k \gets \frac{c - r_{d+1}}{c}$
			
			\lIf{$u \leq k$}{
				$b \gets d-1$
			}
			\lElse{
				$a \gets d+1$
			}
			$d \gets \lfloor (a+b)/2 \rfloor$
		}
		
		$s \gets b+1$
		
		$a \gets s$; 
		$b \gets n$; 
		$d \gets \lfloor (a+b)/2 \rfloor$
		
		$u \gets U(0,1)$
		
		\While{$a \leq b$}{
			$k \gets \frac{(d-s+1)x_s - w_s + w_{d+1}}{c_s}$
			
			\lIf{$u \leq k$}{
				$b \gets d-1$
			}
			\lElse{
				$a \gets d+1$
			}
			$d \gets \lfloor (a+b)/2 \rfloor$
		}
		$t \gets b+1$ \label{alg:samplingend}
		
		$\tau \gets \textsc{RandomSP}(G, s,t)$ 
		
		$\sample \gets \sample \cup \{ \tau \}$

	}
	
	\Return $S$
\end{algorithm2e}

\begin{proof}[Proof of Proposition~\ref{prop:runningtime}]
	Since the percolation states are sorted in non-increasing order, 
	as $x_s-x_t\geq 0$ for all $s< t$, 
	it is immediate to observe that, after the preprocessing steps (lines~\ref{alg:preprocessingstart}-\ref{alg:preprocessingend}), 
	it holds 
	\begin{align*}
	w_i &= \sum_{j=i}^n x_j , ~~~
	c_i = \sum_{j=i}^n R(x_i - x_j), \\
	r_i &= \sum_{j=i}^n c_j , ~~~
	c = \sum_{u,w : u \neq w} R(x_u - x_w) .
	\end{align*}
	We now prove that Algorithm~\ref{algo:nonuniform_sampling_alg}, at each iteration of the second for loop, draws pairs of nodes $s,t$ with probability $\tilde{\kappa}(s,t)$. 
	First, we prove the algorithm samples $s$ with marginal probability 
	$\sum_{z \in V} \tilde{\kappa}(s,z)$. 
	To do so, it uses a binary search (lines 9-16) using the indices $a,b$ and a uniform random number $u \in [0,1]$ to identify 
	\begin{align*}
	s = \argmin_{x} \left\{ \sum_{i=1}^x \sum_{z \in V} \tilde{\kappa}(i,z) \geq u \right\}. 
	\end{align*}
	First, note that for any $x \in [1,n]$, it holds
	\begin{align*}
	 \sum_{i=1}^x \sum_{z \in V} \tilde{\kappa}(i,z) = \frac{c - r_{x+1}}{c} .
	\end{align*}
	Note that the r.h.s. is the quantify $k$ computed in line~12. 
	Therefore, the loop invariant defined with the conditions 
	\begin{align*}
	\sum_{i=1}^x \sum_{z \in V} \tilde{\kappa}(i,z) < u , \forall x < a , \\ 
	\sum_{i=1}^x \sum_{z \in V} \tilde{\kappa}(i,z) \geq u , \forall x > b ,
	\end{align*}
	easily imply the correctness of the first binary search, i.e., the fact that $s$ is chosen with probability $\sum_{z \in V} \tilde{\kappa}(s,z)$. 
	Similarly, the second binary search (lines 17-24) samples $t$ with probability 
	$\tilde{\kappa}(s,t) / \sum_{z \in V} \tilde{\kappa}(s,z)$ 
	by finding, using an independent uniform random number $u \in [0,1]$, 
	\begin{align*}
	t = \argmin_{x} \left\{ \sum_{i=1}^x \frac{\tilde{\kappa}(s,i)}{\sum_{z \in V} \tilde{\kappa}(s,z)} \geq u \right\}. 
	\end{align*}
	Note that the function within the $\argmin$ is computed in line~20; the correctness of the second binary search follows with an invariant that is analogous of the one defined above. 
	Therefore, each pair of nodes $s,t$ is sampled with probability $\tilde{\kappa}(s,t)$.
	Then,  
	the fact that the procedure $\textsc{RandomSP}(G, s,t)$ returns a shortest path $\tau_{st}$ chosen uniformly at random from the set $\Gamma_{st}$, in time $\bigO(T_{\text{\textsc{BBFS}}})$, implies that every $\tau_{st}$ is sampled i.i.d. from the importance distribution $q$. 
	
	For the time complexity, the preprocessing phase (lines~1-6) needs $\bigO(n)$ time while the sampling phase (lines~7-24) needs $\bigO(\ell(\log n + T_{\text{\textsc{BBFS}}} ))$ time, obtaining the bound in the statement. 
	The space bound follows from observing that Algorithm~\ref{algo:nonuniform_sampling_alg} and \textsc{RandomSP} use linear space in the number $n$ and $m$ of nodes and edges. 
\end{proof}

\begin{proof}[Proof of Lemma~\ref{lemma:sumofpercs}]
First, note that $\sigma_{st}(v) = 0$ 
when $s=v$ or $t=v$;
therefore, from the definition of $p(v)$, we have 
\begin{align*}
\sum_{v \in V} p(v) 
& = \sum_{v \in V} \sum_{\substack{s,t\in V\\ s \neq v \neq t}}{\frac{\sigma_{st}(v)}{\sigma_{st}} \kappa(s,t,v)} \\
& = \sum_{v \in V} \sum_{\substack{s,t\in V\\ s \neq t}}{\frac{\sigma_{st}(v)}{\sigma_{st}} \kappa(s,t,v)} \\
& \geq \sum_{v \in V} \sum_{\substack{s,t\in V\\ s \neq t}}{\frac{\sigma_{st}(v)}{\sigma_{st}} \tilde{\kappa}(s,t)} \\
& = \sum_{\substack{s,t\in V\\ s \neq t}} \tilde{\kappa}(s,t) \sum_{v \in V} {\frac{\sigma_{st}(v)}{\sigma_{st}} } \\
& = \sum_{\substack{s,t\in V\\ s\neq t}} \tilde{\kappa}(s,t) | I( \tau_{st} ) |  = \rho ,
\end{align*}
obtaining the lower bound. 
For the upper bound, following similar steps, we have 
\begin{align*}
\sum_{v \in V} p(v) 
& = \sum_{v \in V} \sum_{\substack{s,t\in V\\ s \neq t}}{\frac{\sigma_{st}(v)}{\sigma_{st}} \kappa(s,t,v)} \\
& \leq \hat{d} \sum_{v \in V} \sum_{\substack{s,t\in V\\ s \neq t}}{\frac{\sigma_{st}(v)}{\sigma_{st}} \tilde{\kappa}(s,t)}  = \hat{d} \rho .
\end{align*}
\end{proof}

\begin{proof}[Proof of Theorem~\ref{thm:sample_size}]
First, we observe that the estimator $\tilde{p}(v)$ is an average of $\ell$ i.i.d. random variables with codomain $[0,\hat{d}]$. 
Then, from the definition of $\hat{v}$ and Lemma~\ref{lemma:variancebound}, it holds 
$\Vard{ \tilde{p}(v) }{q} \leq \min\{ \hat{v} , g(p(v)) \}$. 
For any $v \in V$, Bennet's inequality (Thm.~2.9 of~\cite{boucheron2013concentration}) implies that 
\begin{align*}
&\Pr_{\sample}( | \tilde{p}(v) - p(v) | > \varepsilon ) \\
&\leq 2 \exp \left( - \frac{\ell \min\{ \hat{v} , g(p(v)) \} }{ \hat{d}^2 } h\left( \frac{\hat{d}\varepsilon}{\min\{ \hat{v} , g(p(v)) \}} \right) \right) \\
&= 2 B(\varepsilon , \min\{ \hat{v} , g(p(v)) \} , \ell )  ,
\end{align*}
where 
$B(\varepsilon , x , \ell ) = \exp \left( - \frac{\ell x }{ \hat{d}^2 } h\left( \frac{\hat{d}\varepsilon}{x} \right) \right)$. 
Therefore, from an union bound, it holds 
\begin{align*}
& \Pr_{\sample}( \exists v : | \tilde{p}(v) - p(v) | > \varepsilon ) \\
&\leq \sum_{v \in V} \Pr_{\sample}( | \tilde{p}(v) - p(v) | > \varepsilon ) \\
&\leq \sum_{v \in V} 2B(\varepsilon , \min\{ \hat{v} , g(p(v)) \} , \ell ).
\end{align*}
Note that the values of $p(v)$ are not known; thus, 
to upper bound the sum above, we define the following linear program
\begin{align*}
\max & \sum_{x \in \Q \cap (0,1)} n_x 2B(\varepsilon , \min\{ \hat{v} , g(x) \} , \ell ) \\
 s.t.  & \sum_{x \in \Q \cap (0,1)} n_x x \leq \hat{\rho} , \\
 & 0 \leq n_x \leq \frac{\hat{\rho}}{x}  , ~~ n_x \in \R .
\end{align*}
Observe that, from the union bound above, the optimal objective is an upper bound to the probability $\Pr_{\sample}( \exists v : | \tilde{p}(v) - p(v) | > \varepsilon )$. 
We want to show that, if $\ell$ is chosen as in the statement, this probability is $\leq \delta$, proving the theorem. 
To do so, we follow similar steps of~\cite{Pellegrina_2023}, observing that the LP above is an instance of a Bounded Knapsack problem (after a continuous relaxation). 
By defining $x^\star$ as 
\begin{align*}
x^\star = \argmax_{x \in (0,\hat{x}]} \frac{2B(\varepsilon , g(x) , \ell )}{ x } ,
\end{align*}
we obtain the optimal solution 
$n_{x^\star} = \hat{\rho} / x^\star$, $n_x = 0 ,  \forall x \neq x^\star$, 
with objective 
\begin{align}
\frac{2B(\varepsilon , g(x^\star) , \ell ) \hat{\rho}}{x^\star} . \label{eq:optobj}
\end{align}
We now prove that, if $\ell$ is chosen as in the statement, then \eqref{eq:optobj} is $\leq \delta$, proving the statement. For any $x \in (0 , \hat{x}]$, it holds 
\begin{align}
\frac{2B(\varepsilon , g(x) , \ell ) \hat{\rho}}{x} \leq \delta ~~~ \text{ if } ~~~
\ell \geq \frac{ \hat{d}^2 \ln \left( \frac{2 \hat{\rho} }{x \delta } \right) } { g(x)h \left( \frac{\varepsilon \hat{d}}{g(x) } \right) } ,
\end{align}
which follows from the definition of $\ell$ of the statement. 
\end{proof}

\begin{proof}[Proof of Proposition~\ref{prop:rhohatbound}]
We note that $\tilde{\rho}(\mathcal{S})$ is an average of $\ell$ i.i.d. random variables with codomain $[0 , D]$, where $D$ is an upper bound to the vertex diameter. 
Moreover, it holds $\E_\sample[\tilde{\rho}(\mathcal{S})] = \rho$. 
From the application of an Empirical Bernstein bound (Thm.~4 of~\cite{Maurer_2009}) 
to the average $\tilde{\rho}(\mathcal{S})$, after scaling it by $1/D$,
we have that it holds $\rho \leq \hat{\rho}$ with probability $\geq 1 - \delta$. 
Finally, observe that $\sum_{v \in V} p(v) \leq \hat{d} \rho$ 
from Lemma~\ref{lemma:sumofpercs}. 
\end{proof}

\begin{proof}[Proof of Proposition~\ref{prop:boundmaxvar}]
First, we note that from Lemma~\ref{lemma:variancebound}, for any $v \in V$ it holds $\Vard{\tilde{p}(v)}{q} \leq \hat{d} p(v)$;  
therefore, $\max_v \Vard{\tilde{p}(v)}{q} \leq \max_v \hat{d} p(v)$. 
The proof uses the fact that $\max_v \tilde{p}(v)$ is a sharp empirical estimator of $\max_v p(v)$, as proved in~\cite{pellegrina2023efficient} (see Thm.~4.3). 
In fact, $\max_v \tilde{p}(v)$ is a self-bounding function~\cite{boucheron2013concentration}, and satisfies 
\begin{align*}
\max_v p(v) \leq 
\hat{d} \max_{v \in V} \biggl\{ \tilde{p}(v) + \sqrt{ \frac{2 \tilde{p}(v) \log(1/\delta)}{ \ell } } +
		\frac{ \log(1/\delta)}{3 \ell} \biggr\}
\end{align*}
with probability $\geq 1 - \delta$. 
The statement follows by multiplying both sides by $\hat{d}$. 
\end{proof}

\begin{proof}[Proof of Proposition~\ref{prop:algcorrect}]
The statement follows by combining the guarantees of Theorem~\ref{thm:sample_size} (replacing $\delta$ by $\delta/2$), 
Proposition~\ref{prop:rhohatbound} (replacing $\delta$ by $\delta/4$), 
and 
Proposition~\ref{prop:boundmaxvar} (replacing $\delta$ by $\delta/4$),
observing that all the bounds in the statements are computed by \percis\ in Algorithm~\ref{algo:mainalg}. 
Then, from a union bound, all such statements hold simultaneously with probability $\geq 1 - \delta$. 
\end{proof}

\begin{proof}[Proof of Proposition~\ref{thm:likelihoodpercis}]
For any node $v \in V$, we bound its likelihood ratio $d_v$ as follows:
\begin{align*}
d_v 
& = \max_{s,t : \tilde{\kappa}(s,t) > 0} \frac{\kappa(s,t,v)}{\tilde{\kappa}(s,t)} \\
 & = \frac{\sum_{\substack{(u,w)\in V\times V\\ u \neq w}}R(x_u-x_w)}{\sum_{\substack{(u,w)\in V\times V\\ u \neq v \neq w}}R(x_u-x_w)} \\
 &= 1 + \frac{\sum_{u\in V} |x_v-x_u|}{\sum_{\substack{(u,w)\in V\times V\\ u \neq v \neq w}}R(x_u-x_w)} \\
&\leq 1 + \frac{ n }{\sum_{\substack{(u,w)\in V\times V\\ u \neq v \neq w}}R(x_u-x_w)} \\
&\leq 1 + \frac{ n }{ (n-3)\Delta + \Delta } = 1 + \bigO(1) .
\end{align*}
In the last step we used the following argument:
let $s,t$ be the
two nodes $s \neq v \neq t$ with $x_s \geq x_t + \Delta$; then, the sum in the denominator can be lower bounded by the difference of the percolation states of the other $n-3$ nodes with $s$ and $t$, which is at least $\Delta$, plus $x_s - x_t \geq \Delta$. 
The statement follows from the fact that $\hat{d} = \max_v d_v$, and that each $d_v \in \bigO(1)$. 
\end{proof}

\begin{proof}[Proof of Proposition~\ref{thm:likelihoodunif}]
Consider a graph $G = (V,E)$ composed of $n \geq 4$ nodes, with three nodes $a,b,c$ 
with percolation states $x_a = 1$, $x_b = 0$, $x_c = 1/2$, 
and all other nodes with percolation state $0$. 
First, note that $\Delta = 1/2$ and $\hat{d} = \max_v d_v \geq d_b$. 
It holds
\begin{align*}
d_b 
& = \max_{s,t \in V} \kappa(s,t,b) n(n-1) \\
& \geq \kappa(a,c,b) n(n-1) \\
& = \frac{n(n-1)}{\sum_{\substack{u,w\in V \\ u \neq b \neq w}}R(x_u-x_w)}  \\
& = \frac{n(n-1)}{ 3/2(n-3) + 1/2 }  \in \Omega(n) .
\end{align*}
\end{proof}

\begin{proof}[Proof of Proposition~\ref{thm:uniflowerbound}]
Consider a directed graph $G = (V,E)$ composed of $n \geq 4$ nodes, with three nodes $a,b,c$ 
with percolation states $x_a = 1$, $x_b = 0$, $x_c = 1/2$, 
which form a linear directed path, i.e., 
$(a,b) , (b,c) \in E$; 
  all other $n-3$ nodes in $H = V \setminus \{a,b,c\}$ are strongly connected and have percolation state $0$,
  and there exist an edge from a node in $H$ to the node $a$. 
Note that $G$ is weakly connected, and that $\Delta = 1/2$. 
Furthermore, it holds 
$p(v) = 0, \forall v \in V \setminus \{ b \}$, and $p(b) > 0$.
Set $\varepsilon = p(b)/2$. 
We first prove the lower bound $\Omega(n^2)$ to the number of samples needed by \unif, which is based on the uniform sampling distribution. 
Define $\tilde{p}(b)$ the approximation returned by \unif\ after drawing $\ell$ random samples; $\tilde{p}(b)$ is defined as the (weighted) faction of shortest paths where $b$ is internal. 
Note that 
the only shortest path that traverse $b$ is $a ,b ,c$. 
To guarantee $| \tilde{p}(b) - p(b) | \leq \varepsilon$ it is necessary that the pair of nodes $(a,c)$ is sampled at least once by the algorithm;
otherwise, $\tilde{p}(b) = 0$ and $| \tilde{p}(b) - p(b) | > \varepsilon$. 
Define $X$ as a random variable that models the number of random samples that are drawn until the pair $(a,c)$ is sampled by \unif; 
then $X$ is geometrically distributed with mean $n(n-1)$,
thus $\ell \in \Omega(n^2)$; the first part of the statement holds.

We now prove the upper bound $\mathcal{O}(n)$ to the number of samples needed by \percis\ to achieve an $\varepsilon$-approximation.
To do so, we prove bounds to $\hat{d}$, $\hat{\rho}$, and $\hat{v}$ and apply Theorem~\ref{thm:sample_size}. 
From Proposition~\ref{thm:likelihoodpercis} it holds $\hat{d} \in \mathcal{O}(1)$.
Then, it holds $\hat{\rho} = p(b)$ and $\hat{v} = p(b)(\hat{d}-p(b))$, with
\begin{align*}
p(b) = \kappa(a,c,b) = \frac{1/2}{n-3 + 1/2} = \Theta \left( \frac{1}{n} \right) .
\end{align*}
Therefore, after plugging these bounds to the sample size $\ell$ from  Theorem~\ref{thm:sample_size}, it holds $\ell \in \mathcal{O}(n)$; the statement follows.
\end{proof}

\subsection{Missing Experiments}
In this section we show all the experiments that have been omitted from the main paper.


\begin{table}[htb!]
\renewcommand{\arraystretch}{1.2} 
\centering
\caption{Graphs used in our evaluation, where $|V|$ denotes the number of nodes, $|E|$ the number of edges, $D$ the exact diameter, $\rho$ the average number of internal nodes (type D stands for directed and U for undirected).}\label{tab:datasets}
\begin{tabular}{lccccc}
\bottomrule
\multicolumn{1}{l}{\textbf{Graph}} & $\bm{|V|}$ & $\bm{|E|}$ & $\bm{D}$ &$\bm{\rho}$ & \textbf{Type} \\ \hline
{P2P-Gnutella31}                         & 62586          & 147892         & 31          &       7.199          & D             \\
{Cit-HepPh}                         & 34546          & 421534         & 49              &         5.901      & D             \\
{Soc-Epinions}                           & 75879          & 508837         & 16              &       2.755      & D             \\
{Soc-Slashdot}                           & 82168          & 870161         & 13              &     2.135       & D             \\
{Web-Notredame}                     & 325729         & 1469679        & 93             & 9.265         & D             \\
{Web-Google}                        & 875713         & 5105039        & 51             &   9.713       & D             \\
{Musae-Facebook}                          & 22470          & 170823         & 15              &     2.974            & U             \\
{Email-Enron}                              & 36692          & 183831         & 13              &   2.025         & U             \\
{CA-AstroPH}                         & 18771          & 198050         & 14                 &   2.194         & U             \\
\bottomrule 
\end{tabular}
\end{table}

\begin{table}[htb!]

	\renewcommand{\arraystretch}{1.2} 
	\centering
	\caption{The Labeled Networks considered in our experiments. 
	$\mathcal{L}$ indicates the label type (binary or real values) and $\mathcal{L}_{\text{avg}}$ the average values of the labels. }\label{tab:datasets_labeled}
	\begin{tabular}{lcccccc}
		\bottomrule
		\multicolumn{1}{l}{\textbf{Graph}} & $\bm{|V|}$ & $\bm{|E|}$ & $\bm{\mathcal{L}_{\text{avg}}}$ &$\bm{\mathcal{L}}$&$\bm{\rho}$ &\textbf{Type}\\ \hline
		{Guns}                         & 632659          & 5741968         & 0.347               &  $\{0,1\}$    &      2.859          & U   \\
		{Combined}                         & 677753          & 6134836         & 0.246                  & $\{0,1\}$      &         3.053        & U \\
		{Youtube}                           & 152582          & 6268398         & 0.310                  & $[0,1]$  &       2.563          & D   \\
		\bottomrule 
	\end{tabular}
\end{table}

\paragraph{Maximum Error for the IC setting}
Figure~\ref{fig:sd_component} shows the Maximum Error of the approximation computed by {\sc PercIS} with a fixed sample size, under the IC setting. 
More precisely, the plots compare the approximated and exact percolation centrality scores for the $50$ nodes in the isolated component, 
where the estimates are 
returned by, respectively, \percis\ and \unif. 
We focus on the largest graph (see Table~\ref{tab:datasets}), and use $\ell = 10^6$ samples. 

From the plots it is clear that 
{\sc PercIS} returns accurate estimates of the percolation centrality scores for all the nodes in the isolated community, whereas \unif\ fails to provide any meaningful estimation (as all the estimates are equal to $0$). 
This experiment highlights the advantage of our importance-based sampling distribution over uniform sampling, particularly in large real-world graphs that contain small isolated communities where an infection or information spread may originate. 
In such cases, \unif\ misses key nodes and fails to compute reliable estimates.

\begin{figure*}[htb!]
\centering
	\includegraphics[width=\linewidth]{./img/experiments/shared_legend}
	\begin{subfigure}{0.24\textwidth}
		\includegraphics[width=\linewidth]{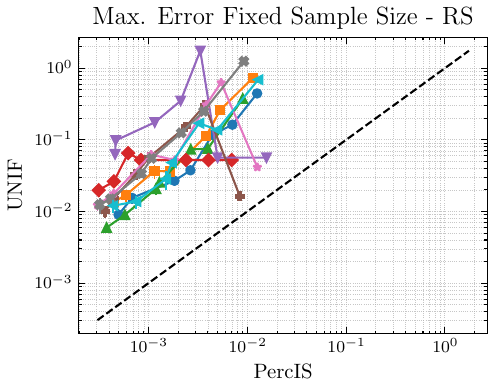}
		\caption{}\label{fig:sd_rnd_init}
	\end{subfigure}
	\begin{subfigure}{0.24\textwidth}
		\includegraphics[width=\linewidth]{./img/experiments/avg_SD_spread}
		\caption{}\label{fig:sd_spread}
	\end{subfigure}
	\begin{subfigure}{0.24\textwidth}
		\includegraphics[width=\linewidth]{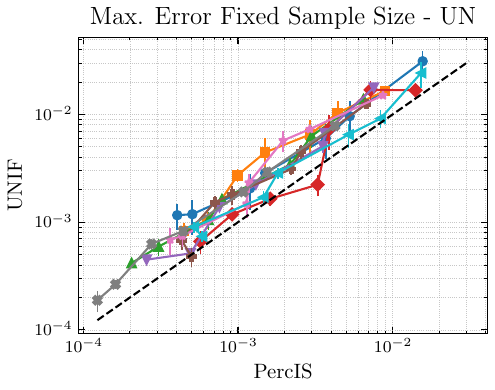}
		\caption{}\label{fig:sd_unif}
	\end{subfigure}
	\begin{subfigure}{0.12\textwidth}
	\end{subfigure}
	\caption{Maximum Errors of \percis\ ($x$ axes) and \unif\ ($y$ axes) on random samples of fixed sizes $\ell \in [10^3, 10^6 ] $.}\label{fig:sd_fixed_ss_appendix}
\end{figure*}

\begin{figure*}[htb!]
\centering
	\begin{subfigure}{0.24\textwidth}
	\center
		\includegraphics[width=\linewidth]{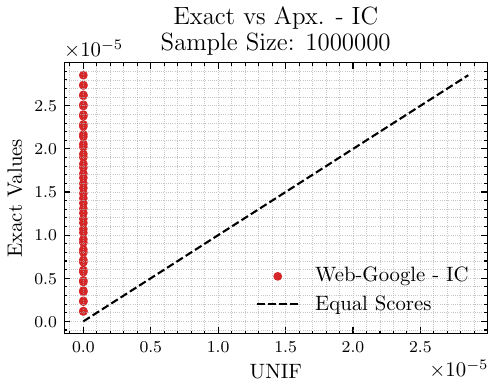}
		\caption{}\label{fig:google_uni}
	\end{subfigure}
	\begin{subfigure}{0.24\textwidth}
	\center
		\includegraphics[width=\linewidth]{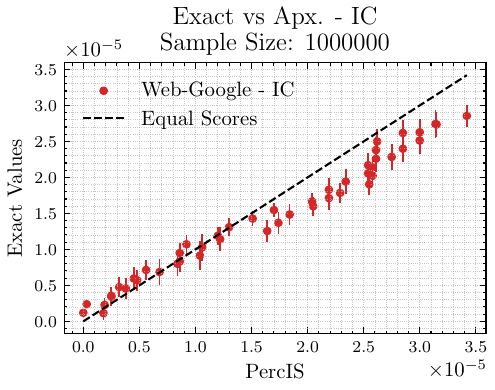}
		\caption{}\label{fig:google_non_uni}
	\end{subfigure}
	\caption{Comparison of exact ($y$ axes) and approximated ($x$ axes) centrality scores computed by \unif\ (a) and \percis\ (b), for the $50$ nodes that belong to the isolated component of the IC setting, for the Web-Google graph and $\ell = 10^6$ samples.}\label{fig:sd_component}
\end{figure*}

\paragraph{Mean Absolute Error}
We evaluate the Average Absolute Error (AE) of the estimates computed by \percis\ and \unif\ from samples of fixed size.
The results are shown in Figure~\ref{fig:mae_fixed_ss}. 
The plots report the AEs for the RS, RSS, and UN settings over all nodes and multiple sample sizes.

As expected, the AE of \percis\ decreases rapidly and consistently as the sample size increases. 
While the behavior of the AE of \unif\ was similar in the UN settings, 
in other cases we observed a much different trend, in particular for the RS setting: in such instances, increasing the sample size had a much smaller impact to the average errors.
This confirms that the resources required by \unif\ to achieve small approximation errors significantly exceed the ones needed by \percis, in particular for challenging instances.

\begin{figure*}[htb!]
\centering
		\includegraphics[width=\linewidth]{./img/experiments/shared_legend}
	\begin{subfigure}{0.24\textwidth}
		\includegraphics[width=\linewidth]{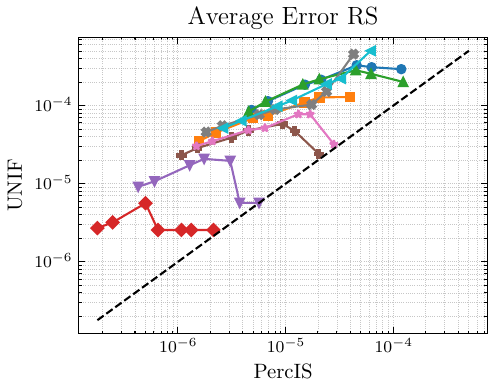}
		\caption{}\label{fig:mae_rnd_init}
	\end{subfigure}
	\begin{subfigure}{0.24\textwidth}
		\includegraphics[width=\linewidth]{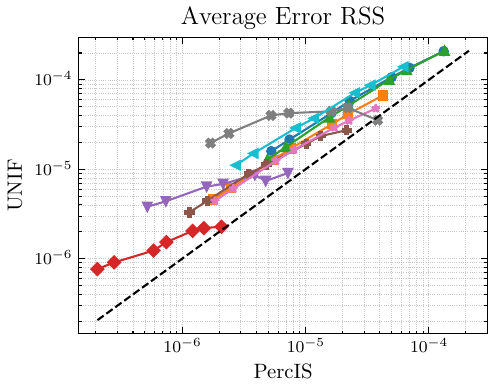}
		\caption{}\label{fig:mae_spread}
	\end{subfigure}
	\begin{subfigure}{0.24\textwidth}
		\includegraphics[width=\linewidth]{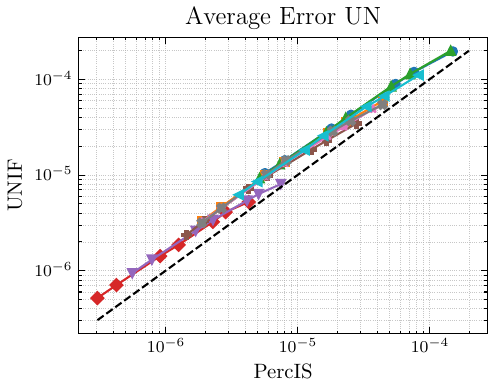}
		\caption{}\label{fig:mae_unif}
	\end{subfigure}
	\caption{Average Absolute Error of \unif\ and \percis\ for fixed sample sizes $\ell \in \{10^3, 5\cdot 10^3,10^4,5\cdot 10^4,10^5,5\cdot 10^5,10^6\}$.}\label{fig:mae_fixed_ss}
\end{figure*}

\paragraph{Missing Plots Uniform and Non-Uniform sampling running times}
Here we compare the time needed by \percis\ and \unif\ using fixed sample sizes. 
Figure~\ref{fig:time_settings} shows this comparison on the RS, RSS, IC, and UN settings. 
We observe that the two methods typically require the same time to analyze the same number of samples; this confirms that the Importance Sampling overhead is minimal. 
Interestingly, we observe a severe speedup for {\sc PercIS} for the IC setting. This improvement is motivated by the fact that our approach always samples a source node from the isolated component, and rapidly completes the BFS towards the target node $t$, as the size of the isolated component is constant.
This is strong contrast with \unif, that instead samples pairs $s,t$ 
that are more expensive to evaluate, and are not useful to obtain accurate approximations.

\begin{figure*}[htb!]
\centering
	\includegraphics[width=\linewidth]{./img/experiments/shared_legend}
	\begin{subfigure}{0.24\textwidth}
		\includegraphics[width=\linewidth]{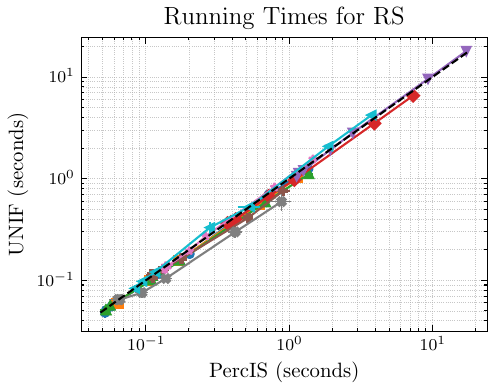}
		\caption{}\label{fig:time_rnd_init}
	\end{subfigure}
	\begin{subfigure}{0.24\textwidth}
		\includegraphics[width=\linewidth]{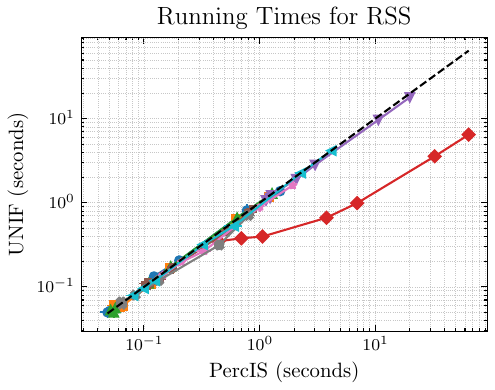}
		\caption{}\label{fig:time_spread}
	\end{subfigure}
	\begin{subfigure}{0.24\textwidth}
		\includegraphics[width=\linewidth]{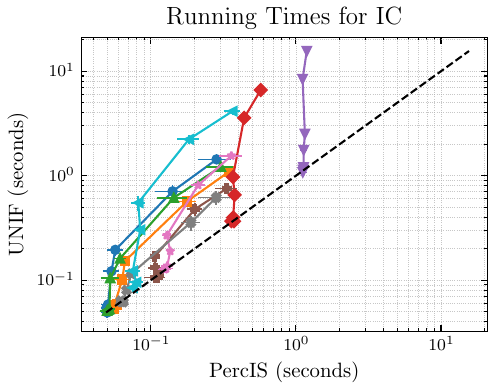}
		\caption{}\label{fig:time_comp}
	\end{subfigure}
	\begin{subfigure}{0.24\textwidth}
		\includegraphics[width=\linewidth]{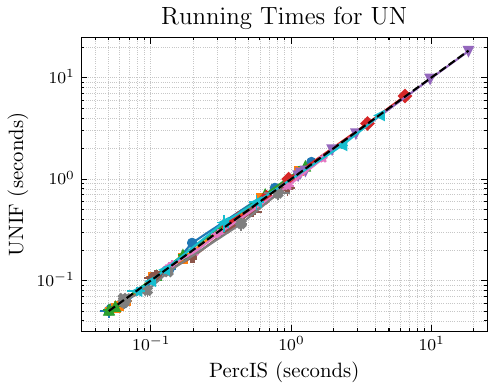}
		\caption{}\label{fig:time_unif}
	\end{subfigure}
	\caption{Comparison between the running times (in seconds) of \unif\ and \percis\ on fixed sample sizes. }\label{fig:time_settings}
\end{figure*}

\paragraph{Missing Plots for upper bound on sample size}
Figure~\ref{fig:ss_vc_vs_rho} shows the comparison between our novel upper bound on the sufficient sample size and a data-independent approach that skips the first phase of {\sc PercIS}. 
We can see that the result in Theorem~\ref{thm:sample_size} provides significant improvements on the number of samples needed to achieve $\varepsilon$-approximations of the percolation centrality.

\begin{figure*}[htb!]
\centering
	\includegraphics[width=\linewidth]{./img/experiments/shared_legend}
	\begin{subfigure}{0.24\textwidth}
		\includegraphics[width=\linewidth]{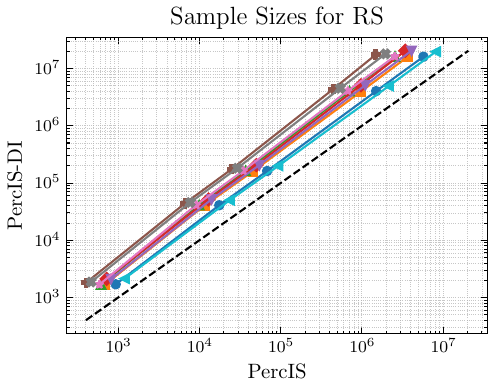}
		\caption{}\label{fig:ss_rnd_init}
	\end{subfigure}
	\begin{subfigure}{0.24\textwidth}
		\includegraphics[width=\linewidth]{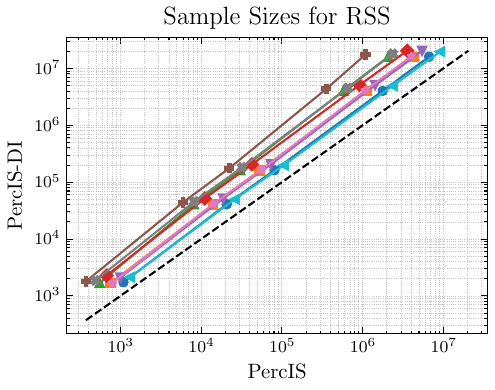}
		\caption{}\label{fig:ss_spread}
	\end{subfigure}
	\begin{subfigure}{0.24\textwidth}
		\includegraphics[width=\linewidth]{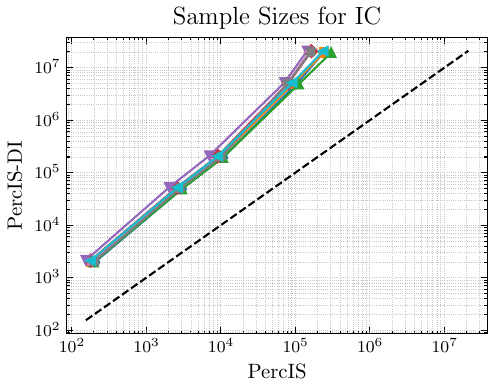}
		\caption{}\label{fig:ss_comp}
	\end{subfigure}
	\begin{subfigure}{0.24\textwidth}
		\includegraphics[width=\linewidth]{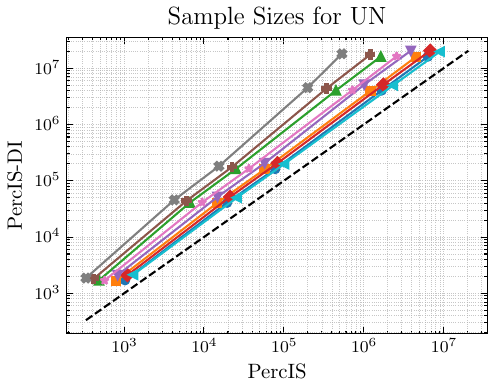}
		\caption{}\label{fig:ss_unif}
	\end{subfigure}
	\caption{Comparison between the number of samples needed by \percis-DI (that skips the first phase of the algorithm and uses a Data Independent bound) and \percis.}\label{fig:ss_vc_vs_rho}
\end{figure*}

\begin{figure*}[htb!]
\centering
	\begin{subfigure}{0.24\textwidth}
		\includegraphics[width=\linewidth]{./img/experiments/bar_plot_times_rnd_init}
		\caption{}\label{fig:bfs_sampling_time_rnd_init}
	\end{subfigure}
	\begin{subfigure}{0.24\textwidth}
		\includegraphics[width=\linewidth]{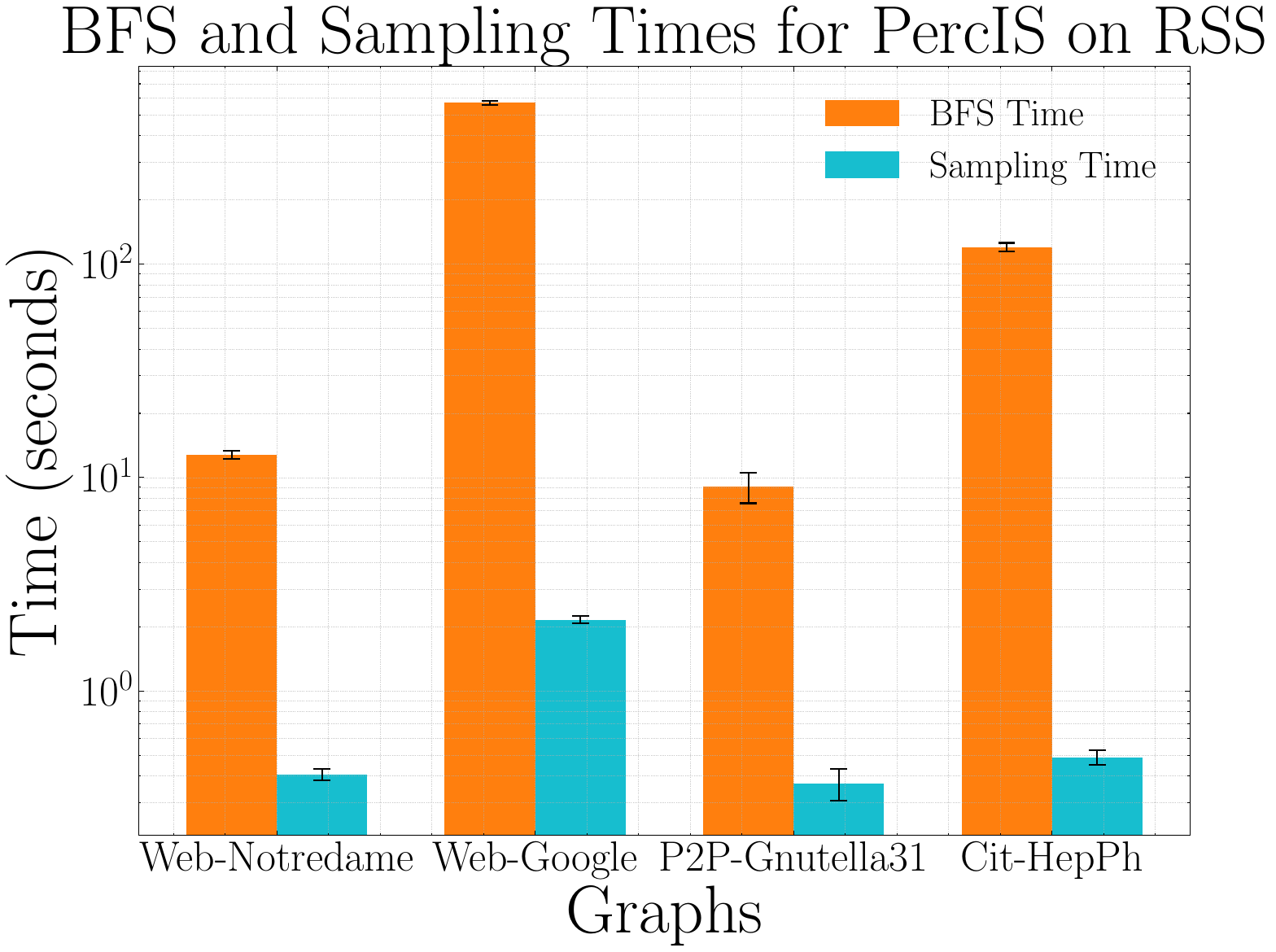}
		\caption{}\label{fig:bfs_sampling_time_time_spread}
	\end{subfigure}
	\begin{subfigure}{0.24\textwidth}
		\includegraphics[width=\linewidth]{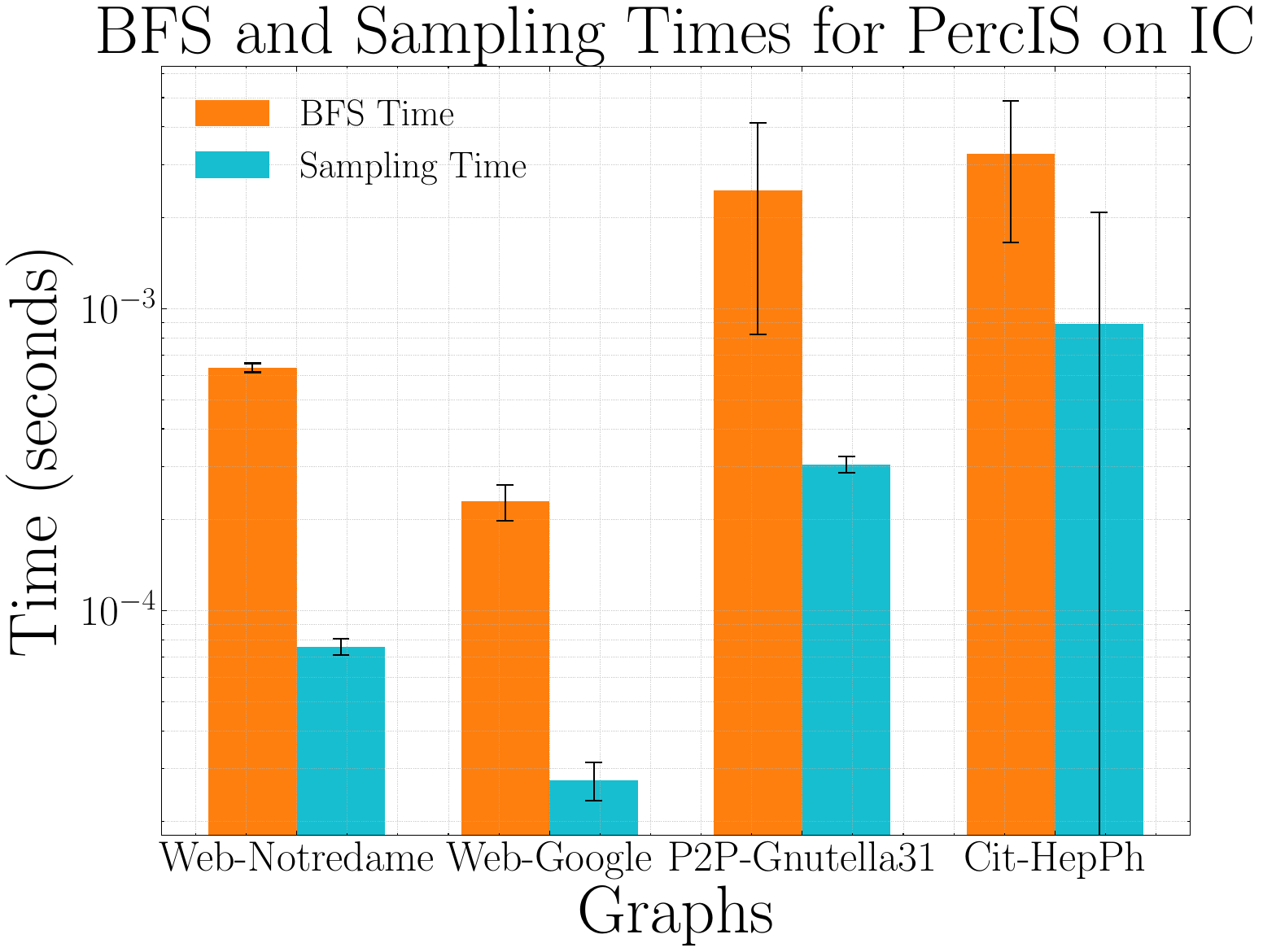}
		\caption{}\label{fig:bfs_sampling_time_time_comp}
	\end{subfigure}
	\begin{subfigure}{0.24\textwidth}
		\includegraphics[width=\linewidth]{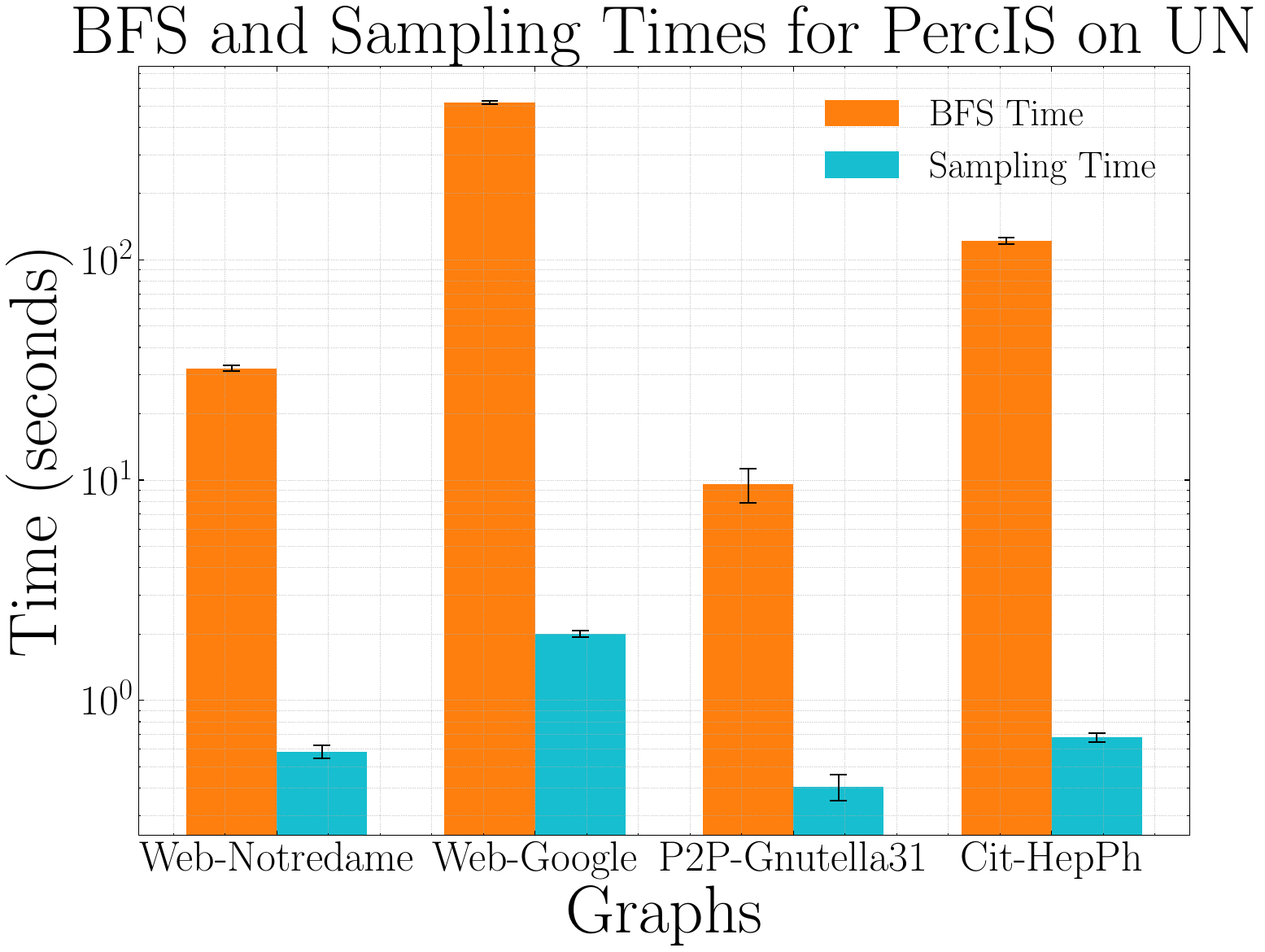}
		\caption{}\label{fig:bfs_sampling_time_time_unif}
	\end{subfigure}
	\caption{Comparison between the total time (in seconds) needed by the bidirectional BFSs and to sample $\ell$ pairs of nodes $(s,t)$ using Algorithm~\ref{algo:nonuniform_sampling_alg}, for $\ell = 10^6$ samples.}\label{fig:bfs_sampling_time_appendix}
\end{figure*}

\paragraph{Experiments for Labeled Networks}
Here we show the experiments for LN that have been omitted from the main paper. 
Table~\ref{tab:datasets_labeled} shows the statistics of the labeled graphs considered in our experiments, while Table~\ref{tab:jaccard} show the Jaccard similarity between the top $k\in\{10,50,100\}$ nodes for the betweenness and the percolation centrality. We observe that the top-$k$ nodes are significantly different when obtained with these centrality measures. 

This divergence suggests that percolation centrality captures different structural roles than betweenness. In practical terms, using percolation centrality to identify key bridges for information or infection spreading may yield better results, as it is sensible to the percolation states of the nodes, which are instead ignored by betweenness.
\begin{table}[htb!]
		\renewcommand{\arraystretch}{1.2} 
	\centering
	\caption{Jaccard similarity between the top $k$ nodes for the betweenness centrality and the percolation centrality.}\label{tab:jaccard}
	\begin{tabular}{lccc}
		\toprule
		& \multicolumn{3}{c}{\textbf{Jaccard Similarity Top-K}} \\ \cline{2-4} 
		\textbf{Graph} & \textbf{10}      & \textbf{50}     & \textbf{100}     \\ \hline
		{Guns}       & 0.053            & 0.087           & 0.117            \\
		{Combined}      & 0.0              & 0.031           & 0.015            \\
		{Youtube}        & 0.429            & 0.369           & 0.504     \\ 
		\bottomrule     
	\end{tabular}

\end{table}

\else
\makeatletter
\IfFileExists{appendix_paper.aux}{\@include{appendix_paper.aux}}{}
\makeatother
\fi


\end{document}